\documentclass[%
preprint,
 amsmath,amssymb,
]{revtex4-1}

\usepackage{graphicx}
\usepackage{dcolumn}
\usepackage{bm}
\usepackage{braket}
\usepackage{graphicx}
\usepackage{hyperref}
\usepackage{amsthm}
\usepackage{bbold}
\usepackage{tikz}
\usepackage{amsmath}
\usepackage[shortlabels]{enumitem}
\usepackage{caption}
\usepackage{subcaption}
\usepackage{tkz-euclide}
\usepackage{adjustbox}
\usepackage{tikz}
\usepackage{pst-solides3d}
\usetikzlibrary{matrix}
\usepackage{tikz-3dplot}
\usepackage[colorinlistoftodos,prependcaption,textsize=small]{todonotes}

\theoremstyle{plain}
\newtheorem{thm}{Theorem}[section]
\newtheorem{lem}[thm]{Lemma}
\newtheorem{prop}[thm]{Proposition}
\theoremstyle{definition}
\newtheorem{defn}[thm]{Definition}
\newtheorem{exmp}[thm]{Example}
\theoremstyle{remark}
\newtheorem{rem}[thm]{Remark}

\begin{document}

\preprint{APS/123-QED}

\title{Topological Order from a Cohomological and Higher Gauge Theory perspective}
\thanks{A footnote to the article}

\author{R. Costa de Almeida}
\email{ricardo.costa.almeida@usp.br}
\author{J. P. Ibieta-Jimenez}%
 \email{pibieta@if.usp.br}
 \author{J. Lorca Espiro} 
 \email{j.lorca.espiro@usp.br}
 \altaffiliation[Also at ]{Departamento de Ciencias F\'{i}sicas, Universidad de la Frontera, Avda. Francisco Salazar 01145,
Casilla 54-D Temuco, Chile.}
 \author{P. Teotonio-Sobrinho} 
 \email{teotonio@if.usp.br}
\affiliation{%
 Departamento de F\'{i}sica Matem\'{a}tica, Universidade de S\~{a}o Paulo\\ Rua do Mat\~{a}o Travessa R 187, S\~{a}o Paulo, CEP 05508-090.
}%

\date{\today}

\begin{abstract}
In recent years, attempts to generalize lattice gauge theories to model topological order have been carried out  through the so called $2$-gauge theories. These have opened the door to interesting new models and new topological phases which are not described by previous schemes of classification. In this paper we show that we can go beyond the $2$-gauge construction when considering chain complexes of abelian groups. Based on elements of homological algebra we are able to greatly simplify already known constructions for abelian theories under a single all encompassing framework. Furthermore, this formalism allows us to systematize the computation of the corresponding topological degeneracies of the ground states and establishes a connection between them and a known cohomology, which conveniently characterizes them with a suitable set of quantum numbers.
\end{abstract}

\pacs{Valid PACS appear here}
\maketitle

\tableofcontents

\section{\label{sec:intro}Introduction}

General frameworks for the classification of topological phases of matter are of uttermost importance for both the condensed matter community and the quantum information community. The standard approach has been to study topological phases through Topological Quantum Field Theories (TQFTs) \cite{Witten,Atiyah88} and the classification of its topological phases. For instance, it has been argued that in the absence of any further global symmetry and for the case of gapped phases of finite gauge theories, the Dijkgraaf and Witten classification is appropriate \cite{Witten,Dijkgraaf90}. Being precise, Dijkgraaf-Witten proposed topological actions with gauge group $G$ that have been used to classify gapped phases of matter having $G$ as a global symmetry. This scheme is known as symmetry protected topological (SPT) phases \cite{Wen-Gu09,Chen12,Oshikawa10,Oshikawa12,Fidkowsky11,LevinGu} and is described by the group cohomology classification of SPT phases \cite{Chen-Wen13}. In the same vein, later on a more general classification of SPT phases was given in \cite{Kapustin14}. 

However, other techniques allow to go beyond the Dijkgraaf-Witten approach such as those based on higher gauge fields and higher gauge symmetries. Recently, higher dimensional TQFTs have attracted interest as means to further classify topological phases of matter in $3$D systems and as possible source of quantum error correction codes. Notably, in \cite{Kapustin13}, a class of TQFTs involving $1$-gauge and $2$-gauge fields has been studied, by using $2$-groups instead of the usual notion of group. There, the existence of gapped phases of matter protected by a $2$-group instead of a $1$-group symmetry was suggested. Moreover, along these lines, in \cite{Bullivant16}, a Hamiltonian formulation of the Yetter's homotopy $2$-type TQFT \cite{Yetter} was constructed with the aim of understanding $(3+1)$ topological phases of matter. Related works worth mentioning are those of \cite{Yoshida16}, where bosonic lattice realizations of SPT phases with higher form symmetry are presented, and \cite{Yoshida15,Yoshida17}, where is discussed their connection to fault-tolerant logical gates in topological quantum codes.

In this paper we present a formalism that builds upon the second approach and appears to be a suitable general framework for the study of a large class of models involving higher gauge fields and symmetries in all possible dimensions. We focus on the abelian case and replace the ordinary (abelian) gauge group by a more general mathematical object, namely, a chain complex of abelian groups. As a consequence, the notion of \textit{configuration} is replaced by \textit{maps between two chain complexes}. Moreover, in this framework, a Hamiltonian formulation of higher lattice gauge theories have been defined in arbitrary dimensions, all of whom have a degenerate ground state subspace whose basis elements are in one-to-one correspondence with the cohomology classes having coefficients in the chain complex of abelian groups \cite{Brown}. This formalism also allows to explicitly show that the ground state degeneracy (GSD) is a topological invariant, for which we give a closed formula. Moreover, this setting naturally understands the GSD as exhibiting contributions from each dimension, therefore exhibiting the different intrinsic \textit{topological orders}.  

When taking the first complex related to the geometrical content to be a discretization of a compact manifold, the models describe gapped topological phases of matter \cite{Wen04,Sarma08,Anderson87}. This systematize in a unique framework a large class of intrinsic topological order models \cite{Wen89,Wen90,LevinWen} and profiles this formulation as a prototype for the study of topological phases of matter in higher dimensions, a topic that has been of great interest due to its possible applications in fault tolerant quantum computation settings \cite{Kitaev2,Nayak08,Freedman03}. Concretely, this approach provides possible candidates for quantum memories as quantum error correction codes, as presented in \cite{Hastings} which now becomes a special case of our formulation. The latter can also be understood as a family of quantum CSS stabilizer codes \cite{CalderbankCSS}, i.e. codes that use the ground state subspace to encode quantum information (thus, the number of logical qubits is related to the GSD of the particular model; we refer to \cite{Devitt} for a brief review on quantum codes and related topics). As it is the case with such codes  \cite{Bombin07,Bombin13,Bravyi14}, our class of models can also be studied in terms of Homology, allowing to consider the present as higher dimensional versions of the \textit{homological quantum error correction codes} \cite{Bravyi98,Vrana15,Anderson13hom}. Even further, the main result of this paper allows to characterize the coding space as well as the labels of the logical operators.

The contents of this paper are organized as follows: In Section \ref{sec:math} we establish the mathematical apparatus that will be used to construct the models. We are mostly interested in defining the cochain complex $\text{hom}(C,G)^{\bullet}$ and its dual, where $C$ and $G$ are two chain complexes of abelian groups. In Section \ref{sec:model}, we define the actual models: Hilbert space $\mathcal{H}$ and its states, the operators in question and all their important algebraic relations, and finally the Hamiltonian operator. In Section \ref{sec:gsd} we study the ground state subspace and we state our main result: there is a one-to-one correspondence between the basis elements of the ground state subspace and the elements of the $0$-th cohomology group of the cochain complex $hom(C,G)^{\bullet}$. We then characterize the ground state subspace basis in a way that is suitable to extract appropriate quantum numbers that will represent them. Section \ref{sec:Examples} is comprised of several concrete examples showing how to obtain already available models from the formalism presented here. A second set of examples is devoted to illustrate the calculations of the ground state degeneracy for different topological situations using our main result (Theorem \ref{thm:main}). The final Section \ref{sec:discu} is a wrap-up discussion of this paper in which we stress the main results as well as briefly mention further research topics along the lines presented here. We also connect these results with other related fields of research. At the end of the work, there are four appendices intended to supply additional information to make the present work as self contained as possible.

\section{\label{sec:math}Mathematical Background}

The main goal of this section is to quickly review some notions of homological algebra which will be used in the development of the subsequent sections. For details regarding the subject, we refer the reader to \cite{Weibel}. Additionally, a brief discussion of representations of abelian groups and the corresponding dual groups is presented as these will also play an important role.

\subsection{Review of Homological Algebra}

\begin{defn}
\label{def:ChainCochainComplex}
A chain complex $(C_\bullet, \partial_\bullet)$ is a sequence of abelian groups $\{C_n\}_{n\in\mathbb{Z}}$ and group morphisms $\partial_n:C_n\rightarrow C_{n-1}$ such that the composition of two such morphism is trivial i.e. $\partial_{n-1}\partial_n=0$.

Similarly, a cochain complex $(C^\bullet, d^\bullet)$ is a sequence of abelian groups $\{C^n\}_{n\in\mathbb{Z}}$ and group morphisms $d^n:C^n\rightarrow C^{n+1}$ such that the composition of two such morphism is trivial i.e. $d_{n+1} d_n=0$.
\end{defn}

The development of homological algebra was largely motivated by its applications to algebraic topology. Of particular importance are the homology and cohomology groups defined below.

\begin{defn}
\label{def:HomologyCohomologyGroups}
Given a chain complex $(C_\bullet, \partial_\bullet)$, the homology groups $H_n(C)$ associated to it are defined by
\begin{align*}
H_n(C)=\text{ker}(\partial_n)/\text{im}(\partial_{n+1})
\end{align*}
In a similar fashion, if $(C^\bullet, d^\bullet)$ is a cochain complex then the cohomology groups $H^n(C)$ assigned to it are defined by
\begin{align*}
H^n(C)=\text{ker}(d^n)/\text{im}(d^{n-1})
\end{align*}
\end{defn}

These groups can be used as sources to construct topological invariants. Roughly speaking, given a topological space $X$ with suitable properties, it is possible to assign a chain complex $C(X)$ to it by some discretization procedure and define the homology $H_n(X)$ as $H_n(C(X))$. For instance, any manifold $M$ can be realized as a simplicial complex and there is a standard procedure for building a chain complex using this. Even though there are different simplicial complexes that correspond to the same manifold, all of them have the same homology groups such that $H_n(M)$ is a well-defined topological invariant. We provide an introduction to simplicial complexes and their geometrical properties in Appendix \ref{ap:algtop} so the reader can see a concrete description of this procedure.

There are also ways to assign a cochain complex to a manifold in order to obtain a topological invariant from the cohomology groups. One example of a topological invariant obtained this way is the De Rham Cohomology. Another procedure relies on the fact that, given a chain complex $(C_\bullet, \partial_\bullet)$ and an abelian group $S$, there is a corresponding cochain complex $(C^\bullet=\text{Hom}(C_n,S), d^\bullet)$ with $d^n(f)=f\partial_{n+1}$. The cohomology groups obtained this way, denoted $H^n(C,S)$, are called cohomology groups with coefficients in $S$.

As usual, when introducing an algebraic structure, it is important to define what the correct notion of a morphism is. The usual definition for chain complexes is by means of the chain maps but we choose to start with a more flexible definition and later on define chain maps as a special case.

\begin{defn}
\label{def:MapDegreeP}
Given two chain complexes $(C_\bullet,\partial_\bullet)\,,\,(C'_\bullet,\partial'_\bullet)$ a $p$-map $f:(C_\bullet,\partial_\bullet)\rightarrow(C'_\bullet,\partial'_\bullet)$ is a sequence of morphisms $f_n:C_n\rightarrow C'_{n-p}$.  
\end{defn}

The set of all $p$-maps, denoted $\text{hom}(C,C')^p$, is actually an abelian group under the binary operation defined by $(f+g)_n=f_n+g_n$. The unit of the group is the trivial $p$-map, denoted $0$, defined by the trivial morphisms $0_n: C_n \rightarrow C'_{n-p}$. It is then straightforward to verify that
\begin{align}\label{def:hom}
\text{hom}(C,C')^p = \prod_n \text{Hom}(C_n,C'_{n-p}) \: ,
\end{align}
which diagrammatically can be represented as in Figure \ref{fig:ChainMaps}.

\begin{figure}[h]
\begin{center}
\begin{adjustbox}{max size={.5\textwidth}{.5\textheight}}
\begin{tikzpicture}
  \matrix (m) [matrix of math nodes,row sep=2em,column sep=3em,minimum width=2em]
  {   \cdots & C_{n+1} & C_{n} & C_{n-1} & \cdots  \\
      \cdots  & C'_{n+1} & C'_{n} & C'_{n-1} & \cdots \\};
  \path[-stealth]
    (m-1-1) edge node [above] {$\partial_{n+2} $} (m-1-2)
    		edge node [below] {$ g_{n+2} $} (m-2-2)
    (m-1-2) edge node [above] {$\partial_{n+1} $} (m-1-3)
            edge node [left] {$f_{n+1}$} (m-2-2)
            edge node [below] {$g_{n-1}$} (m-2-3)
    (m-1-3) edge node [above] {$\partial_{n}$} (m-1-4)
    		edge node [left] {$f_{n}$} (m-2-3)
            edge node [below] {$ g_n $} (m-2-4)
    (m-1-4) edge node [above] {$\partial_{n-1}$} (m-1-5)       
    		edge node [left] {$f_{n-1}$} (m-2-4)
            edge node [below] {$g_{n-1}$} (m-2-5)
    (m-2-1) edge node [below] {$\partial'_{n+2} $} (m-2-2)
    (m-2-2) edge node [below] {$\partial'_{n+1} $} (m-2-3)
    (m-2-3) edge node [below] {$\partial'_{n} $} (m-2-4)
    (m-2-4) edge node [below] {$\partial'_{n-1} $} (m-2-5);
\end{tikzpicture}
\end{adjustbox}
\end{center}
\caption{\label{fig:ChainMaps} Two elements $f \in \text{hom} (C,G)^0$ and $g \in \text{hom} (C,G)^1$}
\end{figure}

The abelian groups $\text{hom}(C,C')^p$ give rise to a cochain complex $(\text{hom}(C,C')^\bullet, \delta^\bullet)$ which we define below:

\begin{defn}
\label{def:InternalHom}
Let $\delta^p:\text{hom}(C,C')^p\rightarrow \text{hom}(C,C')^{p+1}$ be the group morphism defined by
\begin{align*}
(\delta^p f)_n=f_{n-1} \partial_n - (-1)^p \partial'_{n-p} f_n 
\end{align*}
where $f\in\text{hom}(C,C')^p$. It is straightforward to check that $\delta^{p+1} \delta^p=0$ so $(\text{hom}(C,C')^\bullet, \delta^\bullet)$ is a cochain complex. The situation is shown in the following diagram:
\begin{align*}
\cdots\xrightarrow{} \text{hom}(C,C')^p \xrightarrow{\delta^p} \text{hom}(C,C')^{p+1} \xrightarrow{} \cdots \quad .
\end{align*}
\end{defn}

The cohomology groups obtained from $(\text{hom}(C,C')^\bullet, \delta^\bullet)$ are denoted $H^p(C,C')$ and referred to as the cohomology groups of $C$ with coefficients in $C'$. These groups are related to the usual cohomology groups of $C$ with coefficients in $H_{n-p}(C')$ by the following theorem due to Ronald Brown \cite{Brown}: 
\begin{thm}
\label{thm:Brown}Given chain complexes $C$, $C'$ there is an isomorphism
\begin{align*}
\prod_n H^n (C , H_{n-p}(C'))\rightarrow H^p(C , C')
\end{align*}
\end{thm}

Note that an element of $\text{ker}(\delta^0)$ corresponds to a sequence of group morphisms $f_n:C_n\rightarrow C'_n$ such that $(\delta^0 f)_n=f_{n-1} \partial_n - \partial'_{n-p} f_n=0$. This condition is precisely what defines chain maps:

\begin{defn}
\label{def:ChainMap}
Given two chain complexes $(C_\bullet,\partial_\bullet)\,,\,(C'_\bullet,\partial'_\bullet)$ a chain map $f:(C_\bullet,\partial_\bullet)\rightarrow(C'_\bullet,\partial'_\bullet)$ is a sequence of morphisms $f_n:C_n\rightarrow C'_n$ such that $f_n \partial'_{n-1}=\partial_n\ f_{n+1}$.
\end{defn}

Chain maps have the important property that they induce group morphism on the corresponding homology groups. It might happen that different chain maps induce the same morphism on homology groups. When that is the case, such maps are called chain homotopic. More precisely, two chain maps $f,f'$ are homotopic whenever there is some $t\in\text{hom}(C,C')^{-1} $ such that $f'=f+\delta^{-1}$, when that is the case $t$ is called a chain homotopy between $f$ and $f'$. Since $H^0(C,C')=\text{ker}(\delta^0)/\text{im}(\delta^{-1})$ by definition, we see that $H^0(C,C')$ is simply the group formed by homotopy classes of chain maps.

\subsection{Dual Groups and Representations}

All irreducible representations of an abelian group $S$ are one-dimensional and form an abelian group $\hat{S}$ when $S$ is finite. It is easier to understand $\hat{S}$ by observing that any irreducible representation $r$ of $S$ is completely specified by a group morphism $\chi_r : S \rightarrow U(1)$ defined by $\chi_r(g)=\text{Tr}(r(g))$, the character of $r$. The binary operation of $\hat{S}$ corresponds to the point wise multiplication of the corresponding characters, i.e.
\begin{align*}
\chi_{r+r'}(g)=\chi_r(g)\chi_{r'}(g)
\end{align*}

Therefore, one can think of $\hat{S}$ as being the group $\hat{S}=\text{Hom}(S,U(1))$ due to the correspondence $r\leftrightarrow \chi_r$. This perspective makes it straightforward to check that, given a morphism $f$ between finite abelian groups, there is a dual morphism $\hat{f}$ defined by $\rho \mapsto \hat{f}(\rho)=\rho\circ f$, such that:

\begin{prop}
The morphism $\hat{f}$ is a group morphism and $\chi_{\hat{f}(\rho)}=\chi_\rho \circ f$.
\end{prop}

\begin{proof}
The following holds, 
\begin{align*}
\chi_{\hat{f}(\rho)}(h)=\text{Tr}(\hat{f}(\rho)(h))= \text{Tr}(\rho(f(h)))=\chi_\rho (f(h))
\end{align*}
or $\chi_{\hat{f}(\rho)}=\chi_\rho \circ f$ as claimed. It follows also that,
\begin{align*}
\chi_{\hat{f}(\rho+\rho')} = \chi_{\rho+\rho'} \circ f = \\ 
= \chi_{\rho}\chi_{\rho'} \circ f = (\chi_{\rho} \circ f)(\chi_{\rho'}\circ f) = \\ 
= \chi_{\hat{f}(\rho)}\chi_{\hat{f}(\rho')} = \chi_{\hat{f}(\rho)+\hat{f}(\rho')}
\end{align*}
and $\chi_{\hat{f}(0)}=\chi_0 \circ f=1=\chi_0$ which implies $\hat{f}(\rho+\rho')=\hat{f}(\rho)+\hat{f}(\rho')$ and $\hat{f}(0)=0$.
\end{proof}

This technique will be used extensively in order to move freely between groups and representations throughout this paper. For general references we derive the reader to \cite{serre,fulton}.

\section{\label{sec:model}Model Definition}

For the remainder of this article we will use two chain complexes: $\left( C_{\bullet}, \partial^C_{\bullet} \right)$ related to the \textit{geometrical content} of the models, and $\left( G_{\bullet}, \partial^G_{\bullet} \right)$ related to the \textit{group theoretic} content of the models. We will assume the existence of sets $K_n$ such that each group $C_n$ is the free abelian group generated by $K_n$. This is inspired by thinking of the elements of $K_n$ as the $n$-dimensional building blocks of some topological space $K=\sqcup_n K_n$ so elements of $C_n$ are formal sums of such blocks. The morphisms $\partial^C_n$ then illustrate how to glue all of the pieces of $K$ by describing the boundary of some $x\in K_n$ as a formal sum of elements of $K_{n-1}$ as it is classically performed from a simplicial complex perspective (See Appendix \ref{ap:algtop}). Given the previous discussion, we then have:

\begin{defn}[Classical Gauge Configuration]\label{def:classconf}
A classical gauge configuration is an element $f \in \text{hom} (C,G)^0$.
\end{defn}

The latter is basically the statement that there is an assignment of group theoretic degrees of freedom to \textit{each building block} $ K_n \ni x \mapsto f_n \left( x \right) \in G_n$. Diagrammatically, we have:
\begin{figure}[h]
\begin{center}
\begin{adjustbox}{max size={.5\textwidth}{.5\textheight}}
\begin{tikzpicture}
  \matrix (m) [matrix of math nodes,row sep=2em,column sep=3em,minimum width=2em]
  {   \cdots & C_{n+1} & C_{n} & C_{n-1} & \cdots  \\
      \cdots  & G_{n+1} & G_{n} & G_{n-1} & \cdots \\};
  \path[-stealth]
    (m-1-1) edge node [above] {$\partial^C_{n+2} $} (m-1-2)
    (m-1-2) edge node [above] {$\partial^C_{n+1} $} (m-1-3)
            edge node [left] {$f_{n+1}$} (m-2-2)
    (m-1-3) edge node [above] {$\partial^C_{n}$} (m-1-4)
    		edge node [left] {$f_{n}$} (m-2-3)
    (m-1-4) edge node [above] {$\partial^C_{n-1}$} (m-1-5)       
    (m-1-4) edge node [right] {$f_{n-1}$} (m-2-4)
    (m-2-1) edge node [below] {$\partial^G_{n+2} $} (m-2-2)
    (m-2-2) edge node [below] {$\partial^G_{n+1} $} (m-2-3)
    (m-2-3) edge node [below] {$\partial^G_{n} $} (m-2-4)
    (m-2-4) edge node [below] {$\partial^G_{n-1} $} (m-2-5);
\end{tikzpicture}
\end{adjustbox}
\end{center}
\caption{\label{fig:classconf} A classical configuration $f \in \text{hom} (C,G)^0$.}
\end{figure}

Therefore, we can now define:
\begin{defn}[Hilbert Space]\label{def:HilbertSpace}
We define the Hilbert space $\mathcal{H}$ as
\begin{align*}
& \mathcal{H}:=\bigotimes_n \bigotimes_{x \in K_n} \mathbb{C} \left[ G_n \right]_x.
\end{align*}
where $\mathbb{C} \left[G_n\right]_x$ is the group algebra of $G_n$ associated to the building block $x \in K$.
\end{defn}

In other words, each Hilbert space $\mathbb{C}[G_n]_x$ is formed from $G_n$ by taking linear combinations of $g \in G_n$, so a general state can be written as:
\begin{align*}
\ket{\phi}=\sum_{g\in G_n} \phi(g) \ket{g} \, ,
\end{align*}
where $\{\ket{g}\}$ is an orthonormal basis of $\mathbb{C} \left[ G_n \right]_x$. Consequently, by Definition \ref{def:classconf}, any configuration is of the form:
\begin{align}\label{Hilbertbasis}
\ket{f}=\bigotimes_n \bigotimes_{x \in K_n} \ket{f_n \left( x \right)},
\end{align}
where now $\left\{ \ket{f} \right\}$ is an orthonormal basis of $\mathcal{H}$. Thus, any general state $\ket{\psi} \in \mathcal{H}$ can be written as the linear combination: 
\begin{align*}
\ket{\psi}=\sum_f \psi(f)\ket{f} \quad \text{with} \quad  f \in \text{hom} (C,G)^0 .
\end{align*}
Summing up, each map $f\in\text{hom}(C,G)^0$ corresponds to a classical configuration of a generalized gauge field, a \textit{higher gauge field}, while the states $\ket{\psi} \in \mathcal{H}$ correspond to quantum field configurations living on the space associated to $K$.

\begin{rem}\label{rem:finiteHS}
In order to avoid technical complications related to infinite dimensional Hilbert spaces we will only consider the case where $K_n$ is a finite set and only non-empty for a finite number of $n$'s. Additionally, all groups $G_n$ must be finite. These assumptions ensure that $\text{hom}(C,G)^p$ is a finite group for all $p$ with order given by:
\begin{align*}
|\text{hom}(C,G)^p|=\prod_n |\text{Hom}(C_n,G_{n-p})|=\prod_n |G_{n-p}|^{|K_n|}
\end{align*}
Particularly, the dimension of the Hilbert space is $ \text{dim}(\mathcal{H}) = | \text{hom} (C,G)^0 | < \infty $.
\end{rem}

\subsection{Operators}

For the rest of the paper we will denote by $\text{hom} (C,G)_p$ the dual groups of $\text{hom} (C,G)^p$. We begin by defining the elementary operators which are generalized versions of the well known \textit{Quantum Double Models} counterparts in their abelian version (See \cite{Kitaev2,Beigi,Andruskiewitsch,Takeuchi,Kitaev1,LevinWen,Bombin1,Mesaros,QDM} for an account on these topics).

\begin{defn}[Shift and Clock operators]\label{def:PQoperators}
Let $\ket{f} \in \mathcal{H}$ with $f \in \text{hom} \left( C, G \right)^0$. Given  $t \in \text{hom}\left( C, G \right)^0$ and $m \in \text{hom} \left( C, G \right)_0$, we define the operators:
\begin{align*}
P_t \ket{f} := \ket{f + t}, \quad \quad Q_m \ket{f} := \chi_m \left( f \right) \ket{f}
\end{align*}
called the \textit{shift} and \textit{clock} operators, respectively.
\end{defn}

It is straightforward to prove that, for $t, t_1 \text{ and } t_2 \in \text{hom}\left( C, G \right)^0$ and $m, m_1 \text{ and } m_2 \in \text{hom} \left( C, G \right)_0$, the relations:
\begin{align*} 
P_{t_1} P_{t_2} = P_{t_1 + t_2}&, \quad \quad Q_{m_1} Q_{m_2} = Q_{m_1 + m_2}, \\ Q_m & P_t = \chi_{m} \left( t \right) P_t Q_m.
\end{align*}
are satisfied.

\begin{rem}
Notice that an \textit{operator perspective only} study is possible in this formalism by means of the trivial map $0 \in \text{hom} \left( C, G \right)^0$ which can be used to write any $\ket{f} \in \mathcal{H}$ in the form $\ket{f} = P_f \ket{0}$.
\end{rem}

\begin{defn}[Generalized gauge transformations and holonomy measure]\label{def:AopBop}
For all  $t\in \text{hom}(C,G)^{-1}$ and $\ket{f}\in \mathcal{H}$, we define the generalized gauge transformation operator $A_t: \mathcal{H} \rightarrow \mathcal{H}$ as:
\begin{align*}
A_t \ket{f} : = \ket{f + \delta^{-1} t}, \; \text{  or equivalently } \; A_t := P_{\delta^{-1} t}.
\end{align*}
 Analogously, for all  $m \in \text{hom} (C,G)_{1}$ and $\ket{f}\in \mathcal{H}$, we define the holonomy measure operator $B_m: \mathcal{H} \rightarrow \mathcal{H}$ as:
\begin{align*}
B_m \ket{f} : = \chi_{\delta_1 m} \left( f \right) \ket{f}, \;  \text{  or equivalently } \; B_m := Q_{\delta_1 m}.
\end{align*}
\end{defn}

Moreover, the previous operators are easily shown to satisfy the relations:
\begin{align}
\nonumber A_t A_{t'}=A_{t+t'}=A_{t'} A_t \, \; & \; \, B_m B_{m'}=B_{m+m'}=B_{m'} B_m  \\
A_t B_m & = B_m A_t \quad , \label{ABalgebra}
\end{align}
see appendix \ref{ap:algrel} for a proof. Furthermore, we can define: 
\begin{defn}[Gauge equivalence]\label{def:gaugeequiv}
Let $\ket{f}, \ket{g} \in \mathcal{H}$. These states are said to be \textit{gauge equivalent} if there exists a $t\in \text{hom}(C,G)^{-1}$ such that $\ket{f}=A_t\ket{g}=\ket{g+\delta^{-1} t}$. 
\end{defn}

This relation is easily proven to define an \textit{equivalence class}, which  will play an important part in describing and characterizing the elements of the ground state subspace. We will leave this discussion for later sections while we continue discussing the formalism.

By using the operators defined in \ref{def:AopBop}, we can further construct:
\begin{defn}[Generalized projector operators]\label{def:ObOp}
Given $s \in \text{hom}(C,G)_{-1}, \, v\in \text{hom}(C,G)^1$ we define
\begin{align*}
\mathcal{A}_s :=& \dfrac{1}{|\text{hom}(C,G)^{-1}|}\sum_{t}\chi_s(t)A_t, \\
\mathcal{B}_v :=& \dfrac{1}{|\text{hom}(C,G)_{1}|}\sum_{m}\chi_m(v)B_m,
\end{align*}
where $t\in \text{hom}(C,G)^{-1}$ and $m \in \text{hom}(C,G)_1$.
\end{defn}
which are shown to satisfy the following relations:
\begin{lem}\label{lemma:AlgABOb}
For all $s,s' \in \text{hom}(C,G)_{-1}$ and $v,v'\in \text{hom}(C,G)^1$, the following relations for the generalized projector operators (Definition \ref{def:ObOp}) are satisfied:
\begin{enumerate}[(i)]
\item Pairwise commutation: 
\begin{align*}
\mathcal{A}_s \mathcal{B}_v &=\mathcal{B}_v \mathcal{A}_s.
\end{align*}

\item Orthogonality: 
\begin{align*}
\mathcal{A}_s \mathcal{A}_{s'} = \delta(s,s') \mathcal{A}_s, \quad  \mathcal{B}_v \mathcal{B}_{v'} = \delta(v,v') \mathcal{B}_v,
\end{align*}
where $\delta(\cdot,\cdot)$ is the Kronecker delta

\item Completeness:
\begin{align*}
\sum_s \mathcal{A}_s = \mathbb{1}, \quad \quad \sum_v \mathcal{B}_v = \mathbb{1},
\end{align*}
where $\mathbb{1}$ is the identity operator.
\end{enumerate}
\end{lem}

Out of which the claim of being projector operators is justified. These projectors will show to be very helpful when studying both the ground state subspace, which we present here, and the excited state subspace.

\begin{rem}\label{rem:A0B0}
We highlight two particular cases of the aforementioned generalized projectors of Definition \ref{def:ObOp}:
\begin{enumerate}[(i)]
\item The generalized projector $\mathcal{A}_0 := \mathcal{A}_{s|s=0}$ , written as the sum:
\begin{align}\label{eq:projA0}
\mathcal{A}_0=\dfrac{1}{|\text{hom}(C,G)^{-1}|}\sum_t A_t, \quad t \in \text{hom}(C,G)^{-1} \;,
\end{align}
which projects any state $\ket{f} \in \mathcal{H}$ into an equal sum of gauge equivalent states and consequently, can be used to characterize two states $\ket{f},\ket{g}\in \mathcal{H}$ as being gauge equivalent if $\mathcal{A}_0\ket{f}=\mathcal{A}_0\ket{g}$, which is evident from Definition \ref{def:gaugeequiv}.

 \item The generalized $\mathcal{B}_0 := \mathcal{B}_{v|v=0}$ , written as the sum:
\begin{align}\label{eq:projB0}
\mathcal{B}_0=\dfrac{1}{|\text{hom}(C,G)_{1}|}\sum_m B_m, \quad m \in \text{hom}(C,G)_{1} \;,
\end{align}
which projects such states $\ket{f} \in \mathcal{H}$ that satisfy $f  \in \text{ker} (\delta^0)$. A relation which is easily obtained by using its definition and the characters orthogonality relations.
\end{enumerate}
\end{rem}

Concerning the generalized projectors of Definition \ref{def:ObOp} and the elementary operators of Definition \ref{def:PQoperators}, the following relations are easily proven to be satisfied:
\begin{lem}\label{lemma:algOPPQ}
Let $f \in \text{hom}(C,G)^0$ and $m\in \text{hom}(C,G)_0$ be arbitrary elements. It holds:
\begin{enumerate}[(i)]
\item $\quad P_f\mathcal{B}_v= \mathcal{B}_{v+\delta^0 f}P_f$, where $v \in \text{hom}(C,G)^1$.
\item $\quad \mathcal{A}_s Q_m = Q_m \mathcal{A}_{s+\delta_0 m}$, where $s \in \text{hom}(C,G)_{-1}$.
\end{enumerate}
\end{lem}
It will be seen later in section \ref{sec:gsd} that the above relations will help with the characterization of the ground state subspace $\mathcal{H}_0\subset \mathcal{H}$.

\subsection{Local Operators and Dynamics}

In order to define a local Hamiltonian, and therefore the dynamics of these models, we can use the projectors of Definition \ref{def:ObOp} to write the corresponding \textit{local gauge transformations} and \textit{holonomy measurements} in a somewhat natural way.  Hence, we need local maps in $\text{hom}(C,G)_{-1}$ and $\text{hom}(C,G)^1$ that are only non trivial at basis elements $x \in K_n$ and their neighborhoods:
\begin{defn}[Localized maps]\label{def:locmaps}
Let $x\in K_n$ and $g \in G_{n-p}$, $r \in \hat{G}_{n+p}$. The local maps $gx^* \in \text{hom}(C,G)^{p}$ and $rx_* \in \text{hom}(C,G)_{-p}$ are defined by:
\begin{align*}
(gx^*) (y)  :=& \begin{cases} g, \text{ if } x=y \\ 0, \text{ otherwise } \end{cases}   \text{and} \\ \quad (rx_\ast)(f) :=& r(f_n(x)).
\end{align*}
\end{defn}

Appendix \ref{ap:localdec} contains supplementary material regarding these local maps and some important technical results. The locality of these maps is evident since:
\begin{align}
\nonumber \chi_{s} \left( g x^* \right) & = \chi_{s_n (x)} \left( g \right) \, , \;\; \, g \in G_{n-p} \,, s_n \left( x\right) \in \hat{G}_{n-p} \\ 
\chi_{r x_*} \left( v \right) & = \chi_{r} \left( v_n (x) \right)  \, , \;\; r \in \hat{G}_{n-p} \, , v_n \left( x \right) \in G_{n-p} \label{pairings}
\end{align}
which has been obtained using the expansions \ref{texp} and \ref{sexp}. Hence, we can define:

\begin{defn}[Local projector operators]\label{def:AxBx}
Given $x\in K_n,g \in G_{n-1}$, $r \in \hat{G}_{n+1}$. We define the \textit{local gauge projector} and \textit{local holonomy projector} as:
\begin{align*}
A_{x}^{r} &:= \mathcal{A}_{r x_*}=\frac{1}{|G_{n+1}|}\sum_{g\in G_{n+1}} \chi_r \left( g \right) A_{gx^*} \; \quad \text{and} \\  
B_{x}^g &:=\mathcal{B}_{gx^*}=\frac{1}{|G_{n-1}|}\sum_{r\in \hat{G}_{n-1}} \chi_{r} \left( g\right) B_{rx_*} \;,
\end{align*}
where the last equality follows directly from the definition of local maps and \ref{AtBsexp2}.
\end{defn}

To make the connection with known models evident, note that the operator $A_{gx^*}$ in the above definition is the one that performs local gauge transformations with parameter $g\in G_{n+1}$ around $x\in K_n$. On the other hand, the operator $B_{rx_*}$ is diagonal with eigenvalues obtained by applying the representation to the value of the $n$-th fake-holonomy along the boundary of $x\in K_n$. Moreover, it's straight forward to obtain the following relations:

\begin{enumerate}[(i)]
\item Pairwise commutation 
\begin{align*}
A_x^r A_y^{r'} = A_y^{r'} A_x^r \quad &; \quad B_x^g B_y^{g'} =B_y^{g'} B_x^g; \\ 
A_x^r B_y^g & = B_y^g A_x^r.
\end{align*}

\item Orthogonality
\begin{align*}
A_x^r A_x^{r'} = \delta(r,r')A_x^r \quad , \quad B_y^g B_y^{g'} =\delta( g, g') B_x^{g},
\end{align*}
where $\delta \left( \cdot,\cdot \right)$ is the Kronecker delta.

\item Completeness
\begin{align*}
\sum_{r \in \hat{G}_{n+1}} A^r_x = \mathbb{1}, \quad \quad \sum_{g \in G_{n-1}} B^g_x = \mathbb{1},
\end{align*}
where $\mathbb{1}$ is the identity operator.
\end{enumerate}

We can see that the local projector operators are mutually commuting and therefore they can be used to define the following Hamiltonian:

\begin{defn}[Hamiltonian Operator]\label{def:Hamiltonian}
We define the Hamiltonian operator $H:\mathcal{H} \rightarrow \mathcal{H}$ to be:
\begin{align*}
H :=-\sum_n \sum_{x\in K_n} A_x^0 -\sum_n \sum_{y\in K_n} B_y^0.
\end{align*}
\end{defn}

This Hamiltonian enforces two kinds of constraints on the ground states. The first kind, related to $A_x^0$, implies ground states must be gauge invariant. The second kind, related $B_x^0$, projects to the trivial fake holonomy sector. As we shall see in the next section, this leads to the topological features of the model and, in particular, to a topological degeneracy for the ground state states.

\section{\label{sec:gsd}Ground States}

We begin by discussing preliminary considerations concerning a suitable description of the ground state subspace $\mathcal{H}_{0} \subset \mathcal{H}$ which, by means of Definition \ref{def:Hamiltonian}, can be characterized by the relation:
\begin{align}\label{eq:GSeq} 
\mathcal{H}_0 := \left\{ \ket{\Psi} \in \mathcal{H} \,|\, {A}^{0}_x \ket{\Psi} = {B}^{0}_x \ket{\Psi} = \ket{\Psi} \right\},
\end{align}
for all $x \in K_n$. Furthermore, it can be easily shown that $\mathcal{H}_0$ is non empty. Nonetheless, as we stressed in the previous section, we can also study the ground state space by using the operators already defined. Moving towards that direction, we define:
\begin{defn}[Ground state projector operator]\label{def:P0}
We call $\Pi_0 : \mathcal{H} \rightarrow \mathcal{H}_0$ the ground state projector operator, defined as:
\begin{align*}
\Pi_0 := \mathcal{A}_{0} \cdot \mathcal{B}_{0} \;,
\end{align*}
where $\mathcal{A}_0 = \mathcal{A}_{s|s=0} $ and $\mathcal{B}_0=\mathcal{B}_{v|v=0}$, are as in Definition \ref{def:ObOp}.
\end{defn}

It is clear that studying the ground state subspace using the characterization (\ref{eq:GSeq}) is equivalent to that of projecting general states using $\Pi_0$. This follows from the equivalent form of the projectors $\mathcal{A}_s$ and $\mathcal{B}_v$ given in Proposition \ref{prop:ObOp}. This last discussion can be formalized in the following Proposition:

\begin{prop}[Ground states]\label{prop:GS}
A state $\ket{\Psi} \in \mathcal{H}$ is a ground state if and only if satisfies
\begin{align*}
\Pi_0 \ket{\Psi} = \ket{\Psi} \;\;,
\end{align*}
where $\Pi_0$ is the ground state projector defined in \ref{def:P0}.
\end{prop}

The latter is the characterization we were looking for and represents the ground states as the eigenvectors of the $\Pi_0$ operator with unitary eigenvalues. However, if a more physical point of view is desired, we can understand $\Pi_0$ as performing the following simultaneous actions over a general state $\ket{f} \in \mathcal{H}$: 
\begin{enumerate}[(i)]
\item  $\mathcal{A}_0$ projects a state of $\mathcal{H}$ into an equal weight sum of gauge equivalent states, as discussed in Remark \ref{rem:A0B0}.

\item $\mathcal{B}_0$ projects to configurations with trivial $n$ fake-holonomy locally. This can be seen when using the decomposition \ref{ap:localdec}, to write:
\begin{align*}
\mathcal{B}_0=\prod_n\prod_{x \in K_n}B^{0}_{x},
\end{align*}
where $B^{0}_{x}$ is as in Definition \ref{def:AxBx}. Therefore, any state $\ket{f}$ is invariant under the action of $\mathcal{B}_0$ only if it is invariant under the action of all $B_x^{0}$ for all $x \in K_n$. 
\end{enumerate}

From here is immediate that:
\begin{prop}
The ground state subspace $\mathcal{H}_0 \subset \mathcal{H}$ in not empty.
\end{prop}

\begin{proof}
We prove this by construction. Applying the operator $\Pi_0$ over the state
\begin{align}\label{def:0G}
\ket{0_G} :=  \mathcal{A}_0 \ket{0} \,\,\,\,\,\, \text{gives} \,\,\,\,\,\, \Pi_0 \ket{0_G} = \ket{0_G},
\end{align}
which, by Proposition \ref{prop:GS}, proves that $\ket{0_G} \in \mathcal{H}_0$. Moreover, this state is always non zero and hence $\mathcal{H}_0 \neq \emptyset$, since $\mathcal{H}_0$ always has at least one element. 
\end{proof}

A comment about the state $\ket{0_G}$ is worth at this point. Notice that by means of Eq. (\ref{eq:projA0}) we can understand it as a superposition of all basis states that are \textit{gauge equivalent} to the state $\ket{0}$. This state will have an important role in the following section.

\subsection{\label{subsec:GSD} Ground State Degeneracy}

Up to now, we have shown that $\mathcal{H}_0$ in non empty. Moreover, we can prove that if this subspace is degenerate, its degeneracy is topological in a sense that will be clear below. Let us begin by proving the following proposition:

\begin{prop}\label{prop:GS1}
Let $f \in \text{hom}(C,G)^0$. The ground state subspace $\mathcal{H}_0 \subset \mathcal{H}$ is composed by states of the form $\ket{f_G} := \mathcal{A}_0 \ket{f}$. Moreover, $\ket{f_G} \in \mathcal{H}_0$ if and only if, $f \in \text{ker} \left( \delta^0 \right)$.
\end{prop}

\begin{proof}
We prove this proposition in two steps:
\begin{enumerate}[(i)]
\item Let us take a general state $\ket{\Psi} \in \mathcal{H}$. A quick calculation shows that:
\begin{align*}
\Pi_0 \ket{\Psi} = \mathcal{B}_0 P_{\Psi} \ket{0_G} \simeq \ket{\Psi} = P_{\Psi} \ket{0}
\end{align*}
by the characterization of the ground state \ref{prop:GS}, where $\ket{0_G}$ has been defined in Eq. (\ref{def:0G}). The latter is only true if $P_{\Psi}$ is proportional to $\mathcal{A}_0$, as it was to be shown.

\item We then take our prototype state to be $\ket{f_G} := \mathcal{A}_0 \ket{f}$. Hence:
\begin{align}\label{eq:Pi0fG}
\Pi_0 \ket{f_G} = \mathcal{B}_0 P_f \ket{0_G} \simeq \ket{f_G} \;.
\end{align}
Now, when using expansion \ref{eq:projB0}:
\begin{align*}
\ket{f_G} & = \dfrac{1}{|\text{hom}(C,G)_{1}|}\sum_s B_s P_f \ket{0_G} \\ 
& =  \dfrac{1}{|\text{hom} (C,G)_{1}|}\sum_s \chi_{\delta_1  s} \left( f \right) P_f B_s  \ket{0_G} \\ 
& = \left(  \dfrac{1}{|\text{hom}(C,G)_{1}|}\sum_s \chi_{ s} \left( \delta^0 f \right) \right) \ket{f_G}
\end{align*}
noticing that the operator $B_s$ is diagonal over $\ket{0_G}$. The Proposition then holds when using the orthonormal relations for the characters, forcing $\delta^0 f = 0$, or equivalently $f \in \text{ker} \left( \delta^0 \right)$.
\end{enumerate}
\end{proof}

We stress that Proposition \ref{prop:GS1} is a complete characterization of $\mathcal{H}_0 \subset \mathcal{H}$, since every configuration is exhausted by the group $f \in \text{hom}(C,G)^0$. However, we can still refine this characterization even further:

\begin{prop}\label{prop:GS2}
The states $\{ \ket{f_G} = \mathcal{A}_0 \ket{f} \; | \; f \in \text{ker} \left( \delta^0 \right) \}$ are a basis for the ground state subspace and are in one-to-one correspondence with elements of $H^0 (C, G)$.
\end{prop}

\begin{proof}
It is clear from Proposition \ref{prop:GS1} that $\ket{f_G} = \mathcal{A}_0 \ket{f} \, \, | \,\, f \in \text{ker} (\delta^0) \,\, \in\mathcal{H}_0$. Also, by Remark \ref{rem:A0B0}, two such states $\ket{f_G}, \ket{g_G}$ satisfy $\ket{f_G} = \ket{g_G}$  if and only if, $f - g \in \text{im} \left( \delta^{-1} \right)$. Hence, the equivalence class $\left[ f \right] \in H^0 (C,G) = \text{ker} \left( \delta^0 \right) / \text{im} \left( \delta^{-1} \right)$ with representative $\ket{f_G}$ is well defined and furthermore, we have a bijection between $\{ \ket{f_G} | f \in \text{ker} \left( \delta^0 \right) \}$ and $H^0 (C,G)$. Projecting the basis $\{ \ket{f} \}$ of $\mathcal{H}$ into a basis of $\mathcal{H}_0$ using the ground state projector $\Pi_0$ leads to $\{ \Pi_0 \ket{f} \; | \; \Pi_0 \ket{f} \neq 0  \} =\{ \ket{f_G} \;| \; \mathcal{B}_0 \ket{f_G} \neq 0 \} = \{ \ket{f_G} \; | \; f \in \text{ker} \left( \delta^0 \right) \} $, and the result follows.
\end{proof}

Proposition \ref{prop:GS2} is quite far reaching, since also implies that for each cohomology class $[f] \in H^0 (C,G)$, there is a well defined operator
\begin{align}\label{Pfcom}
P_{[f]} := \prod_{g \sim f} P_g =P_f \mathcal{A}_0,
\end{align}
where $g \sim f$ denotes gauge equivalence, that creates the ground states from the state $\ket{0}$, this last observation allows us to state the main result of this paper. Given a finite complex $K$, its associated chain complex $\left( C_{\bullet}(K), \partial^C_{\bullet} \right)$, a graded group $\{ G_{\bullet} \}$, its associated chain complex of finite abelian groups $\left( G_{\bullet} , \partial^G_{\bullet} \right)$ and $\mathcal{H}$ the Hilbert space of definition \ref{def:HilbertSpace} with Hamiltonian $H:\mathcal{H} \rightarrow \mathcal{H}$ as defined in \ref{def:Hamiltonian}. The following theorem follows:

\begin{thm}[Dimension of the ground state subspace]\label{thm:main}
The dimension of the ground state subspace $\mathcal{H}_0$ (GSD) is given by: 
\begin{align*}
\text{GSD} = |H^0(C,G)| \cong \prod_n |H^n(C,H_n(G))|
\end{align*}
\end{thm}

\begin{proof}
The proof follows immediately from Propositions \ref{prop:GS1} and \ref{prop:GS2}, as well as Brown's Theorem (See appendix \ref{ap:Brown}).   
\end{proof}

In physical terms, Theorem \ref{thm:main} underscores a very useful way to understand the GSD; this is, there is a contribution to it from each individual cohomology $H^n(C,H_n(G))$. Moreover, intricate relations between \textit{geometrical quantities} (related to the $\left( C_{\bullet} , \partial^C_{\bullet} \right)$ complex) and \textit{gauge quantities} (related to the $\left( G_{\bullet} , \partial^G_{\bullet} \right)$ complex) can be present. Let us exemplify the latter by restricting the geometrical chain as coming from a closed triangulable manifold. This allows us to use the universal coefficient theorem (See \cite{Hatcher} for a general reference), such that we can decompose each contribution as:
\begin{align}
& \label{UCT} H^n \left( C , H_n \left( G \right) \right) = \\ 
\nonumber & = \text{Hom} \left( H_n \left( C \right) ,  H_n \left( G \right) \right) \oplus \text{Ext}^1 \left( H_{n-1} \left( C \right) ,  H_n \left( G \right) \right) \;,
\end{align}
for all $0 \leq n \leq \dim K$. Here is explicit that, if the $\left( G_{\bullet} , \partial^G_{\bullet} \right)$ complex remains unchanged, any two homological triangulable manifolds ($X_1 \cong X_2 \Rightarrow C(X_1) = C(X_2)$) will have the same GSD. On the other hand, the appearance of two different homologies ($H_n \left( C \right)$ and $H_{n-1} \left( C \right)$) in the decomposition (\ref{UCT}) makes the physical interpretation of the terms somewhat cumbersome since it calls for a case by case study.

The theorem also displays its generalization potential, since it is clear that we can move away of strictly geometrical manifolds for the $\left( C_{\bullet} , \partial^C_{\bullet} \right)$ complex. In fact, \textit{mutatis mutandi}, any two objects $\mathfrak{A}$ and $\mathfrak{B}$ in which homology can be defined such that $H_{\bullet} \left( \mathfrak{A} \right) = H_{\bullet} \left( \mathfrak{B} \right)$, will have the same GSD as long as they have the same $\left( G_{\bullet}, \partial^G_{\bullet} \right)$ complex.

\subsection{Characterization of the Ground States}

Theorem \ref{thm:main} gives an efficient way to calculate the ground state degeneracy of a variety of topological models. Nonetheless, Theorem \ref{thm:main} does not give an intuition on the physical properties of the actual states. In this section we intend to remedy this by identifying the operators that measure and distinguish between different ground states, thus enabling us to define suitable \textit{quantum numbers} for the ground states and a framework for their study.

Given the characterization of the ground states discussed in the previous subsection:
\begin{defn}[Measurement operators]\label{def:measOP}
We understand by measurement operators a set of operators $\mathcal{O}:\mathcal{H}\rightarrow \mathcal{H}$ such that $\mathcal{O}\ket{f_G}=\mathcal{O}([f])\ket{f_G}$ where $[f]\in H^0(C,G)$ and $\ket{f_G}=P_{[f]}\ket{0}\in\mathcal{H}_0$ is the corresponding ground state. The different eigenvalues $\mathcal{O}([f])$ are used to distinguish ground states so that two such operators $\mathcal{O},\mathcal{O}'$ are considered equivalent if $\mathcal{O}([f])=\mathcal{O}'([f])$ for all $[f]\in H^0(C,G)$.  
\end{defn}

It should be obvious that a natural choice for such operators in this theory correspond to the clock operators of Definition \ref{def:PQoperators}. However, not any arbitrary clock operator will be an operator that distinguishes between ground states. In fact, we find that the following proposition must hold:

\begin{prop}\label{prop:measureops}
The operator $Q_m:\mathcal{H}\rightarrow\mathcal{H}$ distinguishes between ground states if and only if $m \in \text{ker} (\delta_0 )$.
\end{prop}

\begin{proof}
Consider some $f\in ker(\delta^0)$ and notice that $Q_m\ket{f_G}$ can be written as: 
\begin{align*}
Q_m\ket{f_G}&=\dfrac{1}{|\text{hom}(C,G)^{-1}|}\sum_t Q_m \ket{f+\delta^{-1}t}\\
&=\dfrac{\chi_m(f)}{|\text{hom}(C,G)^{-1}|}\sum_t \chi_{m}(\delta^{-1}t) \ket{f+\delta^{-1}t}\end{align*}
but $Q_m\ket{f_G}$ needs to be proportional to $\ket{f_G}$ and this is true only if $\chi_m(\delta^{-1}t)=1$ for all $t \in \text{hom} (C,G)^{-1}$. Therefore, $m$ must be trivial on $\text{im}(\delta_{-1})$ which is equivalent to $m\in \text{ker}(\delta_0)$. A more direct way to prove this is to observe that $Q_m \mathcal{A}_0=\mathcal{A}_{\delta_0 m} Q_m$ and $\sum_s \mathcal{A}_s=\mathbb{1}$ so $Q_m\ket{f_G}=\chi_m(f)\mathcal{A}_{\delta_0 m}\ket{f}$ can only be proportional to $\ket{f_G}=\mathcal{A}_0\ket{f}$ if $\delta_0 m=0$. 
\end{proof}

Moreover, if $m,m'\in\text{ker}(\delta_0)$ satisfy $m|_{\text{ker}(\delta_0)}=m'|_{\text{ker}(\delta_0)}$ then they define the same measurement operator since: 
\begin{align*}
\chi_{m'}(f)=\chi_{m'|_{\text{ker}(\delta_0)}}(f)=\chi_{m|_{\text{ker}(\delta_0)}}(f)=\chi_{m}(f),
\end{align*}
where we have used the fact that $f\in\text{ker}(\delta^0)$ for $[f]\in H^0(C,G)$. Thus, the measurement operators are defined up to an equivalence class  $[m] \in \text{ker}(\delta_0) /\sim $ defined by the relation $m' \sim  m$ if, and only if, $m'|_{\text{ker}(\delta_0)}=m|_{\text{ker}(\delta_0)}$. Therefore there is a one-to-one correspondence between ground state measurement operators $Q_m$ and elements of $\text{ker} (\delta_0) / \sim$.

Each class $[m]\in \text{ker}(\delta_0)/\sim$ induces a representation $\bar{m}\in \widehat{H}^0(C,G)$ defined by $\bar{m}([f])=m(f)$. To see that this is well defined notice that if $m'\sim m$ then $m'(f)=m(f)$ since $f\in \text{ker}(\delta^0)$ and if $f'\sim f$ then $m(f')=m(f)$ because $f'-f\in \text{im}(\delta^{-1})$. This procedure defines an isomorphism being straightforward to check that the inverse is obtained by mapping $r\in \widehat{H}^0(C,G)$ into the equivalence class $[\tilde{r}]\in \text{ker}(\delta_0)/\sim$ defined by $\tilde{r}(f)=r([f])$. More precisely:

\begin{prop}\label{prop:measureopsmain}
The equivalence classes of operators that distinguish between ground states, $Q_m$, are in one-to-one correspondence with the elements of $\widehat{H}^0(C,G)$.
\end{prop}

This shouldn't come as a surprise. After all, Theorem \ref{thm:main} showed that $H^0(C,G)$ generates the different ground states so it should be expected that the representations of this groups can be used as observables for the ground states. Now, theorem \ref{thm:main} gives the quantum numbers that characterize the ground state subspace $\mathcal{H}_0$. However, the operators in \ref{prop:measureopsmain} relate the measurement of ground states to the representations of the $H^0(C,G)$ group. Then, there is natural way of constructing measuring operators that are labeled by the actual quantum numbers of $H^0(C,G)$:

\begin{prop}[Ground State Measurement Operators]\label{def:measureops}
Let $[f]\in H^0(C,G)$ and $[m]\in \widehat{H}^0(C,G)$. Then:
\begin{align*}
Q^{[f]}=\dfrac{1}{|\widehat{H}^0(C,G)|}\sum_{[m]}\chi_{[m]}([f])Q_{[m]}.
\end{align*}
are the ground state measurement operators of these models.
\end{prop}

\begin{proof}
That the above is a ground state measurement operator can be seen from its action on a ground state $\ket{g_G}\in \mathcal{H}_0$, $[g]\in H^0(C,G)$:
\begin{align*}
Q^{[f]}\ket{g_G}=&\dfrac{1}{|\widehat{H}^0(C,G)|}\sum_{[m]}\chi_{[m]}([f])Q_{[m]}\ket{g_G} \\
=&\dfrac{1}{|\widehat{H}^0(C,G)|}\sum_{[m]}\chi_{[m]}([f])\chi_{[m]}([g]) \ket{g_G} \\
=&\delta([f],[g]) \ket{g_G},
\end{align*}
where in the last line the orthogonality relation of characters was used. This last expression must be contrasted with Definition \ref{def:measOP}, which completes the proof.
\end{proof}

\subsection{Geometrical Interpretation}

Theorem \ref{thm:main} and  Proposition \ref{prop:measureopsmain} suggest that elements of $H^0(C,G)$ can provide suitable quantum numbers to describe the ground state subspace of the model. However, it is unclear what the geometrical interpretation behind some $[f]\in H^0(C,G)$ is, since the operator $P_{[f]}$ performs many operations in different dimensions simultaneously. In order to explore the ground state basis in a more physical way we will make use of the isomorphic quality between the spaces $\prod_n H^n ( C, H_{n-p}(G))$ and $H^p(C , G)$. Concretely, we can use an isomorphism $\alpha^p : \prod_n H^n ( C, H_{n-p}(G)) \rightarrow H^p(C , G)$, as defined in appendix \ref{ap:Brown} (or equivalent), to decompose any $[f]\in H^0(C,G)$ as $[f]=\alpha^p(\sum_n [f_n])$ where $[f_n]\in H^n(C,H_n(G))$, such that:
\begin{align*}
P_{\left[ f \right]} = P_{\sum_n \alpha^p\left[ f_n \right]} = \prod_n P_{\alpha^p\left[f_n \right]}.
\end{align*}

As by the definition of $\alpha^p$, it is clear that the operator $P_{\alpha^p(f_n)}$ only acts non trivially on the $n$-simplexes that intersect $f_n$, exclusively. Therefore, operators such as $P_{\alpha^p(f_n)}$ are the higher dimensional analogous to the string operators used to construct and label the ground states as in the case of the Toric Code \cite{Kitaev1,Kitaev2,Kitaev3,Kitaev4}.

Moreover, due to the fact that $\widehat{H}^0(C;G)=\prod_n \widehat{H}^n(C,H_n(G))$, any measurement operator can be decomposed in a similar fashion as:
\begin{align*}
Q_m = \sum_n Q_{\widehat{\alpha^{-1}} \left( m_n \right)}
\end{align*}
where $m_n\in \widehat{H}^n(C,H_n(G))$. Hence, by direct calculation it follows that $\chi_{m}([f]) = \prod_n \chi_{\alpha^{-1}(m_n)}(\alpha^p[f_n])=\prod_n \chi_{m_n}([f_n])$, which, in turn implies that:
\begin{align*}
Q_m \ket{f_G} = Q_m P_{[f]}\ket{0} = \prod_n \chi_{m_n}([f_n]) \ket{f_G}.
\end{align*}

Therefore, standard cohomological groups provide a source of quantum numbers since any ground state is of the form: 
\begin{align*}
\ket{[f_0]\dots[f_n]} \simeq \prod_n P_{\alpha^p[f_n]}\ket{0}
\end{align*}
and can be specified by observables $\mathcal{O}_{r_n}=Q_{\alpha^p(r_n)}$ obtained from representations $r_n \in H^n(C, H_n(G))$.

\section{\label{sec:Examples} Examples}

Although the construction presented up to here has the potential to define and develop new models with topological order from the general framework, this section is intended to show how it works as an encompassing formalism for a large class of models already found in the literature.

\begin{exmp}[Abelian $1$-gauge: Quantum Double Models]\label{exmpl:AQDM}
Let $X$ be a compact $2$-manifold. We take the complex $\left( C_{\bullet}, \partial^C_{\bullet} \right)$ to be $C_{\bullet} = C_{\bullet} (X)$ with $\partial^C_{\bullet}$ the usual boundary operator. The geometric chain $\left( C_{\bullet}, \partial^C_{\bullet} \right)$ then reads:
        \begin{equation*}
0 \xrightarrow{\partial_{3}} C_{2} \xrightarrow{\partial^C_2}  C_{1}  \xrightarrow{\partial^C_1} C_0 \xrightarrow{\partial^C_0} 0 \; \;. 
 \end{equation*}
We also restrict $\left( G_{\bullet}, \partial^G_{\bullet} \right)$ to be: 
\begin{align*}
\xrightarrow{} 0 \xrightarrow{\partial^G_2} G_{1} \xrightarrow{\partial^G_1} 0 \xrightarrow{} \;\; , 
 \end{align*}
i.e. having only a $1$-gauge. Without loss of generality, taking $C_{\bullet}(X)$ to be square lattice $\mathcal{L}$ embedded in $X$ we then have:
		\begin{eqnarray*}
			\mathcal{H} _{\mathrm{QDM}} = \bigotimes _{x \in K_1} \mathbb{C} \left[ G_{1} \right]_x = \bigotimes _{l \in C(X)} \mathbb{C} \left[ G_{1} \right]_l,
         \end{eqnarray*}
where the last equality follows from recognizing that each $x \in K_1$ is now associated with only one link $l$ of $C(X)$, as expected.

\begin{figure}[!h]
 			\begin{adjustbox}{max size={0.95 \textwidth}{0.95 \textheight}}
            \centering
					\begin{tikzpicture}[
					scale=0.5,
					equation/.style={thin},
					trans/.style={thin,shorten >=0.5pt,shorten <=0.5pt,>=stealth},
					flecha/.style={thin,->,shorten >=0.5pt,shorten <=0.5pt,>=stealth},
					every transition/.style={thick,draw=black!75,fill=black!20}
					]
						\draw[color=blue!20,fill=blue!20] (6,5) rectangle (8,7);
	    				\draw[color=red!20,fill=red!20] (1,2) rectangle (3,4);
						\draw[->, color=gray, ultra thick, >=stealth] (0,0) -- (0,2.2);
						\draw[->, color=gray, ultra thick, >=stealth] (0,2) -- (0,4.2);
						\draw[->, color=gray, ultra thick, >=stealth] (0,4) -- (0,6.2);
						\draw[-, color=gray, ultra thick] (0,6) -- (0,8);
						\draw[->, color=gray, ultra thick, >=stealth] (2,0) -- (2,2.2);
						\draw[->, color=gray, ultra thick, >=stealth] (2,2) -- (2,4.2);
						\draw[->, color=gray, ultra thick, >=stealth] (2,4) -- (2,6.2);
						\draw[-, color=gray, ultra thick] (2,6) -- (2,8);
						\draw[->, color=gray, ultra thick, >=stealth] (4,0) -- (4,2.2);
						\draw[->, color=gray, ultra thick, >=stealth] (4,2) -- (4,4.2);
						\draw[->, color=gray, ultra thick, >=stealth] (4,4) -- (4,6.2);
						\draw[-, color=gray, ultra thick] (4,6) -- (4,8);
						\draw[->, color=gray, ultra thick, >=stealth] (6,0) -- (6,2.2);
						\draw[->, color=gray, ultra thick, >=stealth] (6,2) -- (6,4.2);
						\draw[->, color=gray, ultra thick, >=stealth] (6,4) -- (6,6.2);
						\draw[-, color=gray, ultra thick] (6,6) -- (6,8);
						\draw[->, color=gray, ultra thick, >=stealth] (8,0) -- (8,2.2);
						\draw[->, color=gray, ultra thick, >=stealth] (8,2) -- (8,4.2);
						\draw[->, color=gray, ultra thick, >=stealth] (8,4) -- (8,6.2);
						\draw[-, color=gray, ultra thick] (8,6) -- (8,8);
						\draw[->, color=gray, ultra thick, >=stealth] (-1,1) -- (1.2,1);
						\draw[->, color=gray, ultra thick, >=stealth] (1,1) -- (3.2,1);
						\draw[->, color=gray, ultra thick, >=stealth] (3,1) -- (5.2,1);
						\draw[->, color=gray, ultra thick, >=stealth] (5,1) -- (7.2,1);
						\draw[->, color=gray, ultra thick, >=stealth] (7,1) -- (9.2,1);
						\draw[->, color=gray, ultra thick, >=stealth] (-1,3) -- (1.2,3);
						\draw[->, color=gray, ultra thick, >=stealth] (1,3) -- (3.2,3);
						\draw[->, color=gray, ultra thick, >=stealth] (3,3) -- (5.2,3);
						\draw[->, color=gray, ultra thick, >=stealth] (5,3) -- (7.2,3);
						\draw[->, color=gray, ultra thick, >=stealth] (7,3) -- (9.2,3);
						\draw[->, color=gray, ultra thick, >=stealth] (-1,5) -- (1.2,5);
						\draw[->, color=gray, ultra thick, >=stealth] (1,5) -- (3.2,5);
						\draw[->, color=gray, ultra thick, >=stealth] (3,5) -- (5.2,5);
						\draw[->, color=gray, ultra thick, >=stealth] (5,5) -- (7.2,5);
						\draw[->, color=gray, ultra thick, >=stealth] (7,5) -- (9.2,5);
						\draw[->, color=gray, ultra thick, >=stealth] (-1,7) -- (1.2,7);
						\draw[->, color=gray, ultra thick, >=stealth] (1,7) -- (3.2,7);
						\draw[->, color=gray, ultra thick, >=stealth] (3,7) -- (5.2,7);
						\draw[->, color=gray, ultra thick, >=stealth] (5,7) -- (7.2,7);
						\draw[->, color=gray, ultra thick, >=stealth] (7,7) -- (9.2,7);
						\draw[->, ultra thick, >=stealth] (2,1) -- (2,2.2);
						\draw[->, ultra thick, >=stealth] (2,2.0) -- (2,4.2);
						\draw[-, ultra thick] (2,4.0) -- (2,5.0);
						\draw[->, ultra thick, >=stealth] (0.0,3) -- (1.2,3);
						\draw[->, ultra thick, >=stealth] (1.0,3) -- (3.2,3);
						\draw[-, ultra thick] (3.0,3) -- (4.0,3);
						\draw[->, ultra thick, >=stealth] (6,5.0) -- (6,6.2);
						\draw[-, ultra thick] (6,6.0) -- (6,7.0);
						\draw[->, ultra thick, >=stealth] (8,5.0) -- (8,6.2);
						\draw[-, ultra thick] (8,6.0) -- (8,7.0);
						\draw[->, ultra thick, >=stealth] (6.0,5) -- (7.2,5);
						\draw[-, ultra thick] (7.0,5) -- (8.0,5);
						\draw[->, ultra thick, >=stealth] (6.0,7) -- (7.2,7);
						\draw[-, ultra thick] (7.0,7) -- (8.0,7);
					\end{tikzpicture}
                    \end{adjustbox}
			\caption{\label{cellsQMD} Planar section of a oriented squared lattice $C_{\bullet} (X)$ giving support to the abelian QDM, with the pink and cyan sectors related to the local gauge $A_{gx^*}$ and local holonomy $B_{ry_*}$ operators, respectively.}
                    \end{figure}
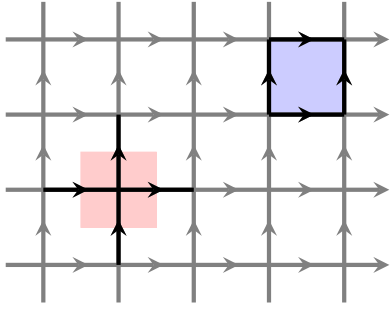
                    
                    \begin{figure}[!h]
 			\begin{adjustbox}{max size={0.95 \textwidth}{0.95 \textheight}}
            \centering
 					\begin{tikzpicture}[scale=0.3,equation/.style={thin},
					trans/.style={thin,shorten >=0.5pt,shorten <=0.5pt,>=stealth},
					flecha/.style={thin,->,shorten >=0.5pt,shorten <=0.5pt,>=stealth},
					every transition/.style={thick,draw=black!75,fill=black!20}]
					\draw[equation] (-9.8,0.0) -- (-9.8,0.0) node[midway,right] {$ A_{gx^*} $};
					\draw[trans] (-6.8,2.3) -- (-6.8,-2.3) node[above=2pt,right=-1pt] {};
					\draw[trans] (-0.1,2.3) -- (1.1,-0.06) node[above=2pt,right=-1pt] {};
					\draw[trans] (-0.1,-2.3) -- (1.1,0.06) node[above=2pt,right=-1pt] {};
					\draw[flecha] (-5.1,0.0) -- (-3.9,0.0) node[above=2pt,right=-1pt] {};
					\draw[flecha] (-4.1,0.0) -- (-1.6,0.0) node[above=2pt,right=-1pt] {};
					\draw[trans] (-1.8,0.0) -- (-0.9,0.0) node[above=2pt,right=-1pt] {};
					\draw[flecha] (-3.0,-1.4) -- (-3.0,-0.5) node[above=2pt,right=-1pt] {};
					\draw[flecha] (-3.0,-0.7) -- (-3.0,1.0) node[above=2pt,right=-1pt] {};
					\draw[trans] (-3.0,0.8) -- (-3.0,1.4) node[above=2pt,right=-1pt] {};
					\draw[equation] (-3.0,2.1) -- (-3.0,2.1) node[midway] {$ a $};
					\draw[equation] (0.6,0.0) -- (0.6,0.0) node[midway,left] {$ b $};
					\draw[equation] (-3.0,-2.1) -- (-3.0,-2.1) node[midway] {$ c $};
					\draw[equation] (-6.7,0.0) -- (-6.7,0.0) node[midway,right] {$ d $};
					\draw[equation] (2.5,-0.07) -- (2.5,-0.07) node[midway] {$ = $};
					\draw[trans] (4.3,2.3) -- (4.3,-2.3) node[above=2pt,right=-1pt] {};
					\draw[trans] (14.6,2.3) -- (15.8,-0.06) node[above=2pt,right=-1pt] {};
					\draw[trans] (14.6,-2.3) -- (15.8,0.06) node[above=2pt,right=-1pt] {};
					\draw[flecha] (7.9,0.0) -- (9.1,0.0) node[above=2pt,right=-1pt] {};
					\draw[flecha] (8.9,0.0) -- (11.4,0.0) node[above=2pt,right=-1pt] {};
					\draw[trans] (11.2,0.0) -- (12.1,0.0) node[above=2pt,right=-1pt] {};
					\draw[flecha] (10.0,-1.4) -- (10.0,-0.5) node[above=2pt,right=-1pt] {};
					\draw[flecha] (10.0,-0.7) -- (10.0,1.0) node[above=2pt,right=-1pt] {};
					\draw[trans] (10.0,0.8) -- (10.0,1.4) node[above=2pt,right=-1pt] {};
					\draw[equation] (10.1,2.0) -- (10.1,2.0) node[midway] {$ a+g $};
					\draw[equation] (15.2,-0.2) -- (15.2,-0.2) node[midway,left] {$ b+g $};
					\draw[equation] (10.1,-2.0) -- (10.1,-2.0) node[midway] {$ c-g  $};
					\draw[equation] (4.4,0.0) -- (4.4,0.0) node[midway,right] {$ d-g  $};
				\end{tikzpicture} 
                \end{adjustbox}
                \caption{\label{kuperberg-figure1} Definition of the local operator $ A_{gx^*}$ with relation to the support regions depicted in FIG. \ref{cellsQMD}. In this case $ x \in K_0$ (i.e. vertices) a regular square lattice $C_{\bullet}(X)$, with $a,b,c,d,g \in G_1$.}
                \end{figure}
                
                \begin{figure}[!h]
 			\begin{adjustbox}{max size={0.95 \textwidth}{0.95 \textheight}}
            \centering
				\begin{tikzpicture}[scale=0.3,equation/.style={thin},
					trans/.style={thin,shorten >=0.5pt,shorten <=0.5pt,>=stealth},
					flecha/.style={thin,->,shorten >=0.5pt,shorten <=0.5pt,>=stealth},
					every transition/.style={thick,draw=black!75,fill=black!20}
					]
					\draw[equation] (-8.8,0.0) -- (-8.8,0.0) node[midway,right] {$ B_{hy_*} $};
					\draw[trans] (-5.6,2.3) -- (-5.6,-2.3) node[above=2pt,right=-1pt] {};
					\draw[trans] (-0.1,2.3) -- (1.1,-0.06) node[above=2pt,right=-1pt] {};
					\draw[trans] (-0.1,-2.3) -- (1.1,0.06) node[above=2pt,right=-1pt] {};
					\draw[flecha] (-4.7,1.2) -- (-2.1,1.2) node[above=2pt,right=-1pt] {};
					\draw[trans] (-2.3,1.2) -- (-0.1,1.2) node[above=2pt,right=-1pt] {};
					\draw[flecha] (-4.7,-1.2) -- (-2.1,-1.2) node[above=2pt,right=-1pt] {};
					\draw[trans] (-2.3,-1.2) -- (-0.1,-1.2) node[above=2pt,right=-1pt] {};
					\draw[flecha] (-3.8,-2.0) -- (-3.8,0.3) node[above=2pt,right=-1pt] {};
					\draw[trans] (-3.8,0.0) -- (-3.8,2.0) node[above=2pt,right=-1pt] {};
					\draw[flecha] (-1.0,-2.0) -- (-1.0,0.3) node[above=2pt,right=-1pt] {};
					\draw[trans] (-1.0,0.0) -- (-1.0,2.0) node[above=2pt,right=-1pt] {};
					\draw[equation] (-2.3,2.0) -- (-2.3,2.0) node[midway] {$ a $};				
					\draw[equation] (0.7,0.0) -- (0.7,0.0) node[midway,left] {$ b $};
					\draw[equation] (-2.3,-2.0) -- (-2.3,-2.0) node[midway] {$ c $};
					\draw[equation] (-5.5,0.0) -- (-5.5,0.0) node[midway,right] {$ d $};
					\draw[equation] (6.7,-0.07) -- (6.7,-0.07) node[midway] {$ = \chi_h \left(-a + b +c -d \right) $};
					\draw[trans] (12.2,2.3) -- (12.2,-2.3) node[above=2pt,right=-1pt] {};
					\draw[trans] (17.7,2.3) -- (18.8,-0.06) node[above=2pt,right=-1pt] {};
					\draw[trans] (17.7,-2.3) -- (18.8,0.06) node[above=2pt,right=-1pt] {};
					\draw[flecha] (13.0,1.2) -- (15.6,1.2) node[above=2pt,right=-1pt] {};
					\draw[trans] (15.4,1.2) -- (17.6,1.2) node[above=2pt,right=-1pt] {};
					\draw[flecha] (13.0,-1.2) -- (15.6,-1.2) node[above=2pt,right=-1pt] {};
					\draw[trans] (15.4,-1.2) -- (17.6,-1.2) node[above=2pt,right=-1pt] {};
					\draw[flecha] (13.9,-2.0) -- (13.9,0.3) node[above=2pt,right=-1pt] {};
					\draw[trans] (13.9,0.0) -- (13.9,2.0) node[above=2pt,right=-1pt] {};
					\draw[flecha] (16.7,-2.0) -- (16.7,0.3) node[above=2pt,right=-1pt] {};
					\draw[trans] (16.7,0.0) -- (16.7,2.0) node[above=2pt,right=-1pt] {};
					\draw[equation] (15.4,2.0) -- (15.4,2.0) node[midway] {$ a $};
					\draw[equation] (18.4,0.0) -- (18.4,0.0) node[midway,left] {$ b $};
					\draw[equation] (15.4,-2.0) -- (15.4,-2.0) node[midway] {$ c $};
					\draw[equation] (12.2,0.0) -- (12.2,0.0) node[midway,right] {$ d $};
				\end{tikzpicture}
                \end{adjustbox}
			\caption{\label{kuperberg-figure2} Definition of the local operator $B_{hy_*}$ with relation to the support regions depicted in FIG. \ref{cellsQMD}. In this case $y \in K_2$ (i.e. plaquetes) of a regular square lattice $C_{\bullet}(X)$, with $a,b,c,d \in G_1$ and $h \in \hat{G}_1$.}
		\end{figure}

        From Definition \ref{def:AxBx}, the realization of the operators of the model over the square lattice $C_{\bullet}(X)$ involves the support regions of FIG. \ref{cellsQMD}, while the operators are defined as those appearing FIG. \ref{kuperberg-figure1}. For completeness, the last identification comes from Definition \ref{def:Hamiltonian}, which gives:
		\begin{eqnarray*}
H_{\mathrm{QDM}} = - \sum_{x\in K_0} A_x^{0} - \sum_{y\in K_2} B_y^{0} = - \sum _{v \in \mathcal{L}} A_{v } - \sum _{p \in \mathcal{L}} B_{p }.
		\end{eqnarray*}
where each $x \in K_0$ is associated to a vertex $v \in \mathcal{L}$ and each $y \in K_2$ is associated to a plaquette $p \in \mathcal{L}$. As it can be seen, the particularizations discussed generate models that coincide with those of the abelian version of the Quantum Double Models.
\end{exmp}

\begin{exmp}[GSD of the Toric code]
From the previous example, let us study the Toric Code, proposed by Kitaev \cite{Kitaev1,Kitaev2,Kitaev3}. For this, we take $X = T^2$ a Torus. Hence, $C_{\bullet} = C_{\bullet} (T^2)$ and $G_1 = \mathbb{Z}_2$. We then have the configurations $f \in \text{hom} \left( C , G \right)^0$ to be of the form shown in figure \ref{TCcomplex} with associated homology groups shown in figure \ref{TCHom}.

\begin{figure}[!h]
 			\begin{adjustbox}{max size={0.95 \textwidth}{0.95 \textheight}}
            \centering
 \begin{tikzpicture}
  \matrix (m) [matrix of math nodes,row sep=2em,column sep=3em,minimum width=2em]
  { 0 & C_{2} & C_{1} & C_{0} & 0 \\
    0  & 0 & \mathbb{Z}_2 & 0  & 0 \\};
  \path[-stealth]
    (m-1-1) edge node [above] {$ \partial^C_3 $} (m-1-2)
    (m-1-2)	edge node [above] {$ \partial^C_2 $} (m-1-3)
            edge node [left] {$ f_2 $} (m-2-2)
    (m-1-3) edge node [above] {$ \partial^C_1 $} (m-1-4)
    		edge node [left] {$f_{1}$} (m-2-3)
    (m-1-4) edge node [above] {$ \partial^C_0 $} (m-1-5)
    		edge node [left] {$ f_0 $} (m-2-4)
    (m-2-1) edge node [below] {$ \partial^G_3 $} (m-2-2)
    (m-2-2) edge node [below] {$ \partial^G_2 $} (m-2-3)
    (m-2-3) edge node [below] {$ \partial^G_1 $} (m-2-4)
    (m-2-4) edge node [below] {$ \partial^G_0 $} (m-2-5);
\end{tikzpicture}
\end{adjustbox}
\caption{\label{TCcomplex} A configuration $f \in \text{hom}\left( C, G \right)^0 $ for the toric code.}
\end{figure}

\begin{figure}[!h]
 			\begin{adjustbox}{max size={0.95 \textwidth}{0.95 \textheight}}
            \centering
\begin{tikzpicture}
\node at (0,0) {$H_n \left( T^2 \right) = \begin{cases}
\mathbb{Z}, \quad \quad n=0 ,\\
\mathbb{Z} \oplus \mathbb{Z}, \, n=1, \\
\mathbb{Z}, \quad \quad n=2.
\end{cases} ; \,\,
H_n (G) = \begin{cases}
0, \quad n=0 ,\\
\mathbb{Z}_2, \; n=1, \\
0, \quad n=2.
\end{cases}$};
\end{tikzpicture}
\end{adjustbox}
\caption{\label{TCHom}  Homology groups of the $\left( C_{\bullet} , \partial^C_{\bullet} \right)$ and $\left( G_{\bullet} , \partial^G_{\bullet} \right)$ complexes associated to the toric code, where $H_n \left( C \right) \cong H_n \left( T^2 \right)$.}
\end{figure}

We immediately use Theorem \ref{thm:main} to calculate the ground state degeneracy of this model:
\begin{align*}
\text{GSD} = |H^0(C,G)| = |H^1(C,H_1(G))| = 2^2,
\end{align*}
where we have used the universal coefficient theorem ( Eq. (\ref{UCT}) ) in order to expand the terms in $|H^0(C,G)|$. The result coincides with the expression for the ground state degeneracy found in the literature.
\end{exmp}

\begin{rem}
We stress the power of the result since for the previous calculation we basically only need the information contained in diagram of Figure \ref{TCcomplex}. 
\end{rem}

\begin{exmp}[GSD of the 3D Toric Code on $T^3$]
In the same vein as the previous example, we now consider the 3$D$ version of the Toric Code on a $3$-torus $T^3$  \cite{Keyserlingk13}. The geometrical complex is $C_{\bullet} = C_{\bullet} (T^3)$ and the $\left( G_{\bullet}, \partial^G_{\bullet} \right)$ complex consists on a single non-trivial group $G_1 = \mathbb{Z}_2$. The homology groups of $T^3$ are given by :
\begin{align*}
H_n (C) \cong H_n(T^3) & =
\begin{cases}
\mathbb{Z}, \quad \quad \quad \quad \; \; n=0,3 \\
\mathbb{Z} \oplus \mathbb{Z} \oplus \mathbb{Z}, \quad n=1,2
\end{cases}
\\
H_n (G) & = \begin{cases}
0, \quad n=0 ,\\
\mathbb{Z}_2, \; n=1, \\
0, \quad n=2.
\end{cases}.
\end{align*}
Using again Theorem \ref{thm:main}, the ground state degeneracy of this model is:
\begin{align*}
\text{GSD}=|H^0(C,G)|=|H^1(C,H_1(G))| \\ 
=|\text{Hom}(H_1(T^3),H_1(G))|=2^3,
\end{align*}
as calculated in \cite{Keyserlingk13}.
\end{exmp}

\begin{exmp}[Abelian $1,2$-gauge theories] \label{exmp:1-2gauge}
We warn the reader that we will sacrifice formality for the sake of keeping the length of the example, and hence this paper, short. We begin by paraphrasing the following results found in \cite{Baez2,Forrester}: Any $2$-group, or equivalently a \textit{crossed module}, completely defines a $2$-gauge theory on a smooth compact manifold $X$ (we refer to \cite{Baez1,Baez2} for definitions and explicit constructions on $2$-groups, since they are outside of the scope of this paper). We focus on the latter since they are easier to connect with our formalism.

We start by considering a compact $n$-manifold such that the complexes $( C_{\bullet} , \partial^C_{\bullet})$ with $C_{\bullet} = C_{\bullet} (X)$ and $\partial^C_{\bullet}$ the usual boundary map, $(G_{\bullet}, \partial^G_{\bullet})$ and a configuration $f \in \text{hom} (C,G)^0$ as depicted in FIG. \ref{examplehgt}. We have taken $G_i = 0$ for all $i \neq 1,2$ so we can readily see that $f_0$ and $f_k$, with $k\geq 3$, do not affect the region of the lower chain between $G_1$ and $G_2$.

\begin{figure}[h]
\centering
\begin{tikzpicture}
  \matrix (m) [matrix of math nodes,row sep=2em,column sep=3em,minimum width=2em]
  {   \cdots & C_{3} & C_{2} & C_{1} & C_{0} & 0  \\
      \cdots & 0 & G_{2} & G_{1} & 0 & 0 \\};
  \path[-stealth]
    (m-1-1) edge node [above] {$\partial^C_{4}$} (m-1-2) 
    (m-1-2) edge node [above] {$\partial^C_{3}$} (m-1-3)
            edge node [left] {$f_3$} (m-2-2)
    (m-1-3) edge node [above] {$\partial^C_{2}$} (m-1-4)
    		edge node [left] {$f_{2}$} (m-2-3)
    (m-1-4) edge node [above] {$\partial^C_{1}$} (m-1-5)
    (m-1-5) edge node [above] {$\partial^C_{0}$} (m-1-6)
    (m-1-5) edge node [right] {$f_0$} (m-2-5)
    (m-1-4) edge node [right] {$f_{1}$} (m-2-4)
    (m-2-1) edge node [below] {$\partial^G_{4}$} (m-2-2)
    (m-2-2) edge node [below] {$\partial^G_{3}$} (m-2-3)
    (m-2-3) edge node [below] {$\partial^G_{2} $} (m-2-4)
    (m-2-4) edge node [below] {$\partial^G_{1}$} (m-2-5)
    (m-2-5) edge node [below] {$\partial^G_{0}$} (m-2-6);
\end{tikzpicture}
\caption{\label{examplehgt}}
\end{figure}

We recognize $(G_{\bullet}, \partial^G_{\bullet})$ to be effectively a \textit{crossed module of groups}, which is defined as the quadruple $( G_2, G_1, \partial^G_2 , \triangleleft)$, where $\triangleleft$ is an action of $G_1$ on $G_2$ by automorphisms; such that (i) $\partial^G_2$ is $G_1$-equivariant: $\partial^G_2 \left( h \triangleleft g \right) = g^{-1} \partial^G_2 \left( h \right) g$ for all $g \in G_1 $ and $h \in G_2 $, and (ii) $\partial^G_2$ satisfies the Peiffer identity: $h' \triangleleft \partial^G_2 \left( h \right) =h^{-1} h' h$ in $G_2$ for all $h, h' \in G_2$. Notice that in our case, the action $\triangleleft$ is innocuous, since the previous requirements are trivialized due to $G_2$ and $G_1$ being abelian groups.

The latter identification relates the chain of FIG. \ref{examplehgt} with a $2$-gauge abelian theory since when a \textit{crossed module} is well defined, a $2$-gauge theory is also well defined via the equivalence through the corresponding $2$-group ( For details see \cite{Baez1,Baez2,Forrester} ). In particular, the latter can be directly related with the class of lattice realizations proposed by Kapustin in Section $4$ of \cite{Kapustin13}. Concretely, our model corresponds to the one obtained by the $2$-group $(G_1,G_2,\partial_2^G,\triangleleft)$ where the action $\triangleleft:G_1\rightarrow \text{Aut}(G_2)$ is trivial.
\end{exmp}

\begin{rem}
The last example shows how the formalism presented in this paper is a general framework for higher gauge theories in their abelian versions.
\end{rem}

\begin{exmp}[GSD of a $\mathbb{Z}_2$, $\mathbb{Z}_4$ Abelian 1,2-gauge theory over a sphere $S^2$] 
We proceed as before by considering a 1,2-gauge theory defined on a discretization of a 2-sphere $S^2$. This is, the $1$-gauge degrees of freedom are located at the $1$-simplices of $C_1 \in C(S^2)$ whereas the $2$-gauge ones live on the $2$-simplices of $C_2 \in C(S^2)$, which is clear from the configurations $f  \in \text{hom}(C,G)^0$. The $\left( G_{\bullet}, \partial^G_{\bullet} \right)$ needs to be defined in a way that the group homomorphisms $\partial^G_{\bullet}$ are compatible. One such possibility is taking $G_1 = \mathbb{Z}_2=\{1,-1\}$ and $G_2=\mathbb{Z}_4=\{1,i,-1,-i\}$, such that group homomorphism is $\partial_2^G(i)=-1$. The configuration $f \in \text{hom}(C,G)^0$ and the homology groups related to this structure are shown in Figs.\ref{12complex} and \ref{12Hom}, respectively.
\begin{figure}[!h]
 			\begin{adjustbox}{max size={0.95 \textwidth}{0.95 \textheight}}
            \centering
\begin{tikzpicture}
  \matrix (m) [matrix of math nodes,row sep=2em,column sep=3em,minimum width=2em]
  { 0 & C_{2} & C_{1} & C_{0} & 0 \\
    0  & \mathbb{Z}_4 & \mathbb{Z}_2 & 0  & 0 \\};
  \path[-stealth]
    (m-1-1) edge node [above] {$ \partial^C_3 $} (m-1-2)
    (m-1-2)	edge node [above] {$ \partial^C_2 $} (m-1-3)
            edge node [left] {$ f_2 $} (m-2-2)
    (m-1-3) edge node [above] {$ \partial^C_1 $} (m-1-4)
    		edge node [left] {$f_{1}$} (m-2-3)
    (m-1-4) edge node [above] {$ \partial^C_0 $} (m-1-5)
    		edge node [left] {$ f_0 $} (m-2-4)
    (m-2-1) edge node [below] {$ \partial^G_3 $} (m-2-2)
    (m-2-2) edge node [below] {$ \partial^G_2 $} (m-2-3)
    (m-2-3) edge node [below] {$ \partial^G_1 $} (m-2-4)
    (m-2-4) edge node [below] {$ \partial^G_0 $} (m-2-5);
\end{tikzpicture}
\end{adjustbox}
\caption{\label{12complex} A configuration $f \in \text{hom}\left( C, G \right)^0 $ for the abelian $1$,$2$-gauge theory.}
\end{figure}

\begin{figure}[!h]
 			\begin{adjustbox}{max size={0.95 \textwidth}{0.95 \textheight}}
            \centering
\begin{tikzpicture}
\node at (0,0) {$H_n \left( S^2 \right) = \begin{cases}
0, \quad \quad n=0 ,\\
0, \quad \quad n=1, \\
\mathbb{Z}, \quad \quad n=2.
\end{cases} ; \, \,
H_n (G) = \begin{cases}
0, \quad n=0 ,\\
0, \quad n=1, \\
\mathbb{Z}_2, \quad n=2.
\end{cases}$};
\end{tikzpicture}
		\end{adjustbox}
\caption{\label{12Hom}  Homology groups for the $\left( C_{\bullet}, \partial^C_{\bullet} \right)$ and $\left( G_{\bullet}, \partial^G_{\bullet} \right)$ associated to the $\mathbb{Z}_2$, $\mathbb{Z}_4$ abelian $1$,$2$-gauge theory. Where $H_n \left( C \right) \cong H_n \left( S^2 \right)$.}
		\end{figure}

It is now immediate to use Theorem \ref{thm:main} to obtain the GSD:
\begin{align*}
\text{GSD}=|H^0(C,G)|=|H^2(C,H_2(G))| \\ 
=|\text{Hom}(H_2(C),H_2(G))|=2.
\end{align*}
Thus, the model exhibits degeneracy when defined on the $2$-sphere $S^2$. 
\end{exmp}

\begin{rem}
Again, this last calculation is made without any modification of the formalism presented. In fact, the GSD of any abelian higher gauge theory with the underlying structure is contemplated by the Theorem \ref{thm:main}.
\end{rem}

\section{\label{sec:discu}Discussion and Outlook}

In this work we have shown that the cohomology $H^p (C,G)$, presented in Section \ref{sec:math}, is a natural structure for the study of the class of models introduced in Section \ref{sec:model}. These models are considered to be higher dimensional generalizations of the abelian QDMs, as shown in Example \ref{exmpl:AQDM}. In fact, in the spirit of \cite{Baez1,Baez2}, with the restrictions discussed in Example \ref{exmp:1-2gauge}, the formalism presented is suitable for abelian models based on a $2$-group structure and models with higher gauge transformations of any order. The main feature of this formalism was discussed in Section \ref{sec:gsd}, where it was proven that the ground state degeneracy of all these models is characterized by $H^0 ( C, G )$ and that the isomorphism $|H^0(C,G)| \cong \prod_n |H^n(C,H_n(G))|$ provides a natural way to characterize their ground state space $\mathcal{H}_0$. A complete set of quantum numbers for these systems have then been obtained following this idea. Unfortunately, even though this characterization is complete, it comes at the stake of a clear physical interpretation. Ground states are now in a one-to-one correspondence with the elements of $ \prod_n |H^n(C,H_n(G))|$ which, in turn, mix the geometrical and the gauge content.

As it was commented at the end of \ref{subsec:GSD}, since it is possible take $\left( C_{\bullet}, \partial^C_{\bullet} \right)$ to be any free finitely generated complex, going away from strictly geometrical supports will produce new models with topological order that are yet to be explored, and for which the formalism presented can be immediately applied without any modification. Furthermore, even if the examples we presented consideri $\left( C_{\bullet} , \partial^C_{\bullet} \right)$ as coming from closed manifolds, the formalism accounts for compact manifolds with boundaries and more general chain complexes not related to geometry. For all those cases the GSD formula of Theorem \ref{thm:main} is still valid. Some examples of such situations will be explored in a forthcoming paper.

In section \ref{sec:Examples} we show how our models successfully reproduce some well known models with topological order. Concerning the higher gauge examples, the relation between the models we obtain in this work and the notion of symmetry protected phases (SPT) is not fully understood, as it was discussed in \cite{Kapustin13}. However, from what is presented in \cite{Yoshida16}, SPT phases with higher global symmetry in $d$ dimensions are shown to have a duality relation with intrinsic topologically ordered phases with local symmetry via a \textit{gauging}. Thus, the models we introduce in this paper are expected to have a duality relation with higher symmetry SPT phases.

Although a detailed analysis of the excited states is reserved for a forthcoming work, a few comments can be made about them. The generalized projector operators of definition \ref{def:ObOp} belonging to $\mathcal{H} / \mathcal{H}_0$ are characterized by the operators of Definition \ref{def:ObOp} for non trivial parameters. This is ensured by means of Lemma \ref{lemma:AlgABOb}. However, since we are considering only abelian models, the higher dimensional analogs of charges ($n$-charges) and fluxes ($n$-fluxes) have trivial statistics between themselves. We expect these same operators to play a fundamental role when studying the bulk-boundary correspondence when $\left( C_{\bullet}, \partial^C_{\bullet} \right)$ comes from a manifold with boundaries.

Another promising way to further study these models is as quantum error correction codes, as presented \cite{Hastings}. Given that the GSD has been completely characterized, the number of potential logical qubits in the code is known. Moreover, the set of logical operators is also completely determined by the quantum numbers given by Theorem \ref{thm:main}. In fact, the distance of the code is also determined by the ground-state mapping operators together with the size of the discretization of the manifold. We expect the duality between the operators mapping between different ground state and their corresponding measuring operators, explicit in this formalism, to help in the classification of the excited states.

Although it is also an open question whether good quantum correcting codes exist for these models \cite{Shor1,Shor2,Steane}, we can say some preliminary things in this respect. It is known that the encoding space of a product code is related to the encoding spaces of the elementary homological codes that compose it via the K\"{u}nneth formula \cite{Bravyi14}, which is a special case of our GSD formula. Moreover, given a product we can readily construct a model that has the same encoding space. The trivial example is to take $\left( G_{\bullet}, \partial^G_{\bullet} \right)$ to be a graded abelian group. In fact, a variety of models can be constructed having the same GSD with different dynamics for the excited state sector just by considering a different $\left( G_{\bullet}, \partial^G_{\bullet} \right)$. In this sense, the homological product codes of \cite{Bravyi14} are special cases of our formalism.


\appendix


\section{\label{ap:algtop}Simplicial Homology and Cohomology}

Let $K$ be a finite simplicial complex. Denote by $K_n\subset K$ the set of $n$-simplexes. An $n$-simplex is a collection of vertices $[v_0, \dots , v_n]$ so that there is an induced orientation given by the lexicographic order $v_0 < v_1 < \dots < v_n$. Therefore, a vertex permutation $p:\{ 0,\dots n\}\rightarrow \{ 0,\dots n\}$ results in the simplex $[v_{p(0)}, \dots , v_{p(n)}]$ which is considered to be the same simplex, with the same orientation if $p$ has even parity and the opposite orientation if $p$ has odd parity. The orientation is taken into account by saying that $ [v_{p(0)}, \dots , v_{p(n)}]=\pm[v_0, \dots , v_n]$ with a plus sign if $p$ has even parity and a negative sign otherwise. We take this orientation to be the default for any simplex $x\in K$ unless otherwise explicitly stated.

\begin{defn}
A simplicial $n$-chain is a formal sum of $n$-simplexes with coefficients in $\mathbb{Z}$: $\quad \sum_{x\in K_n} c_x x $

Moreover, the set formed by simplicial $n$-chains, denoted $C_n$,  is the free abelian group with basis $K_n$.
\end{defn}

\begin{defn}
Let $(C_{\bullet} , \partial^C_{\bullet})$ be: 
\begin{align*}
\cdots \xrightarrow{} C_{n} \xrightarrow{\partial_n^C} C_{n-1} \xrightarrow{} \cdots \quad
\end{align*}
with the maps $\partial^C_n$ defined by:
\begin{align*}
\partial_n^C ([v_0, \dots , v_n])=\sum_{i=0}^{n} (-1)^i [v_0, \dots ,\hat{v_i} , \dots , v_n]
\end{align*}
where $[v_0, \dots ,\hat{v_i} , \dots , v_n]$ denotes the $(n-1)$-simplex obtained by removing the vertex $v_i$ from $[v_0, \dots , v_n]$.
\end{defn}

\begin{prop}
$(C_{\bullet},\partial^C_{\bullet})$ is a chain complex.
\end{prop}

\begin{proof}
It is straight forward to compute:
\begin{align*}
\partial^C_{n-1} \partial^C_n([v_0, \dots , v_n]) = \partial^C_{n-1} \sum_{i=0}^n (-1)^i [v_0, \dots ,\hat{v_i} , \dots , v_n] \\ 
= \sum_{i=0}^{n} (-1)^i \partial^C_{n-1}([v_0, \dots , \hat{v_i} , \dots , v_n])\\
=\sum_{i=0}^n \sum_{j>i} (-1)^{i}(-1)^{j-1}[v_0, \dots ,\hat{v_j} ,\dots,\hat{v_i}, \dots , v_n] + \\ 
+ \sum_{i=0}^n \sum_{j<i} (-1)^{i}(-1)^{j}[v_0, \dots ,\hat{v_i} ,\dots,\hat{v_j}, \dots , v_n]
\end{align*}
Each term $[\dots ,\hat{v_i} ,\dots,\hat{v_j}, \dot{}]$ appears in both sums with opposite signs and therefore $\partial^C_{n-1}\circ\partial^C_n([v_0, \dots , v_n])=0$. The result follows from this. 
\end{proof}

When choosing $H$, an abelian group, we define the cochain complex as below:
\begin{align*}
\cdots \xleftarrow{} Hom(C_{n},H) \xleftarrow{d_{n-1}} Hom(C_{n-1},H) \xleftarrow{} \cdots
\end{align*}
where $d_{n-1}(f) := f  \partial_n^C$. This cochain complex gives rise to the cohomology of $C$ with coefficients in $H$:
\begin{align*}
H^n(C,H)=\text{ker} (d_n)/ \text{im} (d_{n-1})
\end{align*}

In Appendix \ref{ap:Brown} we will be dealing with cohomology with coefficients in the homology groups $H_m(G)$. It will be convenient to denote by $d_n^m:Hom(C_n,H_m(G))\rightarrow Hom(C_{n+1},H_m(G))$ the map $d_n$ with coefficients in $H_m(G)$. Using this notation $H^n(C, H_{n-p}(G))=\text{ker}(d_n^{n-p})/\text{im}(d_{n-1}^{n-p})$. For further details on these topics with focus on simplicial complexes we refer to \cite{Robins} and the references therein.

\section{Isomorphism consequence of Brown's theorem }\label{ap:Brown}

In \cite{Brown} it is shown that the space $H^p (C,G)$ is isomorphic to $\prod_n H^n (C, H_{n-p}(G))$. This result will help on allows us to interpret geometrically the cohomology group $H^0(C,G)$, linked to the ground state subspace. Under these conditions we use the fact that in order to have an isomorphism $\alpha: \prod_n H^n (C; H_{n-p}(G)) \rightarrow H^p(C,G)$, it is sufficient to build an injective morphism. To begin with, let $\left( G_{\bullet} , \partial^G_{\bullet} \right)$ be the chain complex of abelian groups of section \ref{sec:math}. Let $H(G)$ be the chain complex constructed with the homology groups $H_{n}(G)$ of $G$, i.e.,
\begin{align*}
\dots H_{n+1}(G)\xrightarrow{0} H_n(G) \xrightarrow{0} H_{n-1}(G) \dots,
\end{align*}
where the boundary map $0$ is the trivial one. Notice that there is a canonical projection $\pi_n: \text{ker}(\partial_n^G)\rightarrow H_n(G)$. It is then straightforward to check that there are morphisms $\phi_n: H_n(G) \rightarrow \text{ker}(\partial_n^G)\subset G_n$ satisfying the relation $\pi_n \circ \phi_n=\text{id}$. Fix a collection $\phi=\{\phi_n\}_n$ of such morphisms. Next, given $f = \{ f_n: C_n \rightarrow H_{n-p}(G) \}\in \text{hom}(C,H(G))^p$, consider the morphism $\alpha_{\phi}^p: \text{hom}(C,H(G))^p \rightarrow \text{hom}(C,G)^p$ defined by:
\begin{align*}
(\alpha_{\phi}^p(f))_m= \phi_{m-p} \circ f_m \;.
\end{align*}
It is easy to show that the boundary operator in $\text{hom}(C,H(G))$ factors out this morphism, evident when calculating:
\begin{multline*}
(\delta^p (\alpha_{\phi}^p (f)))_m = (\alpha_\phi^p (f))_{m-1} \partial^C_m -(-1)^p \partial^G_{m-p}(\alpha_\phi^p (f))_m = \\ =\phi_{m-1-p} f_{m-1} \partial^C_m -(-1)^p \partial^G_{m-p} \phi_{m-p} f_m = \\ =\phi_{m-1-p}  f_{m-1} \partial^C_m=\phi_{m-1-p} d_{m-1}^{m-1-p}(f_{m-1})=\\=(\alpha_\phi^p (d(f)))_m ,
\end{multline*}
where we have used the fact that the image of $\phi_{m-p}$ is contained in $\text{ker} \left( \partial^G_{m-p} \right)$, i.e. $\partial^G_{m-p} \phi_{m-p}=0$, and $d(f)$ has components $(d(f))_m = d_{m-1}^{m-1-p}(f_{m-1}) = f_{m-1} \partial^C_m$.

It follows from the above discussion that if we choose $f_n$ belonging to $\text{ker} (d_n)$, i.e. $d(f)=0$, their image will also be in the kernel of the boundary operator, i.e., $\alpha_\phi^p(f)\in \text{ker}(\delta^p)$. This is guaranteed by $\delta^p(\alpha_\phi^p (f))=\alpha_\phi^p (d(f))=0$. Moreover, if we pick $g\in \text{hom}(C,H(G))^{p-1}$ then: 
$$\alpha_\phi^p (f)+\alpha_\phi^{p-1}( d(g))=\alpha_\phi^p (f)+\delta^{p-1}(\alpha_\phi^{p-1} (g)).$$

In particular, the morphism preserves the equivalence classes in the cohomology groups. Correspondingly, we can map $\sum [f_n] \in \prod_n H^n (C; H_{n-p}(G))$ into $[\alpha_\phi^p (f)]\in H^p(C;G)$ in a well defined manner by only choosing a representative $f_n$ from each cohomology class $[f_n]$. We formalize the latter as:

\begin{defn}\label{def:alpha}
The morphism $\alpha_\phi^p:\prod_n H^n (C; H_{n-p}(G)) \, \rightarrow \, H^p(C ; G)$:
\begin{align*}
\alpha_\phi^p(\sum [f_n]) := [\alpha_\phi^p (f)],
\end{align*}
where $\sum [f_n]$ is a general element of $\prod_n H^n (C; H_{n-p}(G))$ with $[f_n]\in H^n (C; H_{n-p}(G))$ and $f\in \text{hom}(C,H(G))^p$ consists of a collection of maps $f_n:C_n\rightarrow H_{n-p}(G)$ with $f_n \in [f_n]$. 
\end{defn}

This map is well defined, since it is clear that it does not depend on the choice of representative morphisms. To see this let us take $f'\in \text{hom}(C,H(G))$ such that $f'_n \in [f_n]$. This is, there is some $g\in \text{hom}(C,H(G))^{p-1}$ such that $f'=f+dg$ , hence:
\begin{align*}
[\alpha_\phi^p (f')]=[\alpha_\phi^p (f) +\alpha_\phi^p (d (g))] 
=&[\alpha_\phi^p (f) + \delta^{p-1}(\alpha_\phi^{p-1} (g))]\\
=&[\alpha_\phi^p (f)].
\end{align*}

It remains to prove that $\alpha_\phi^p$ is an injective morphism. To do this it is sufficient to show that the kernel $\text{ker}(\alpha_\phi^p)$ is trivial:

\begin{prop}[injectivity of $\alpha_\phi^p$]
The morphism $\alpha_\phi^p$, defined in \ref{def:alpha},  is injective.
\end{prop}
\begin{proof}
Take $\sum [f_n]$ such that their image under the morphism belongs to the trivial class, namely, $\alpha_\phi^p(\sum [f_n]) = [\alpha_\phi^p (f)] = [0]$, In particular, this means that there is some $t \in \text{hom} (C,G)^{p-1}$ such that $\alpha_\phi^p(f) = \delta^{p-1}(t)$. Therefore,
\begin{multline*}
(\alpha_\phi^{p} (f))_m = (\delta^{p-1}(t))_m \quad \implies \\ 
\phi_{m-p}f_m=t_{m-1} \partial^C_m -(-1)^{p-1}\partial^G_{m-p} t_m.
\end{multline*}

Notice that each image of each term is contained in $\text{ker}(\partial^G_{m-p})$, so when composing with $\pi_{m-p}$ from the left, we get $f_m=\pi_{m-p}t_{m-1} \partial^C_m$, where we have used the relations $\pi_{m-p} \phi_{m-p}=id$ and $\pi_{m-p}\partial^G_{m-p}=0$. Therefore, $f_m=d_{m-1}^{m-p}(\pi_{m-p} t_{m-1})$, equivalently, $[f_m]=[0]$. We now conclude that $\sum [f_n]=[0]$, hence, $\text{ker}(\alpha_\phi^p)$ is trivial. This is, $\alpha_\phi^p$ is injective. as claimed.
\end{proof}

We highlight that even though we made a choice of $\phi$ when defining $\alpha_\phi^p$ the actual morphism does not depend on it. To see this let us choose another collection $\psi$ obeying $\pi_m \circ \psi_m = id$. Then, $(\alpha_\phi^p (f) -\alpha_\psi^p (f))_m=(\phi-\psi)_{m-p} f_m$ but $\pi_m \circ (\phi-\psi)_m=0$, this is, there is some $\xi$ such that $(\phi-\psi)_m=\partial^G_m \xi_m$, therefore:
\begin{multline*}
(\alpha_\phi^p (f) -\alpha_\psi^p (f))_m = \partial^G_{m-p} \xi_{m-p} f_m \\ 
= (-1)^p\left( \xi_{m-1-p} f_{m-1} \partial^C_m -(-1)^{p-1} \partial^G_{m-p} \xi_{m-p} f_m  \right) \\
= \delta^{p-1}((-1)^p \alpha_\xi^p (f))_m,
\end{multline*}
where we have assumed that $f_{m-1}\in \text{ker}(d_{m-1}^{m-1-p})$ and consequently $f_{m-1} \partial^C_m=0$. The above argument implies that $[\alpha_\phi^p (f)]=[\alpha_\psi^p(f)+\delta^{p-1}(\alpha_\xi^p (f))]=[\alpha_\psi^p (f)]$, equivalently $\alpha_\phi^p=\alpha_\psi^p$. We formalize this in the following Proposition:

\begin{prop}
There is a well-defined injective morphism $\alpha^p:\prod_n H^n (C, H_{n-p}(G))\rightarrow H^p(C,G)$ given by $\alpha^p = \alpha_\phi^p$ for some choice of $\phi$. Moreover, this morphism is actually an isomorphism since it is an injective morphism between isomorphic groups.
\end{prop}

\section{Local Decomposition}\label{ap:localdec}

We notice that the group $\text{Hom}(C_n,G_{n-p})$ is equivalent to $ G_{n-p}^{|K_n|}=\{ f: K_n\rightarrow G_{n-p}\}$ since $K_n$ is a basis for $C_n$. Moreover, for $x, \, y \in K_n$ and $g\in G_{n-p}$ one can define 
\begin{align*}
&gx^*:K_n \rightarrow G_{n-p} \\
&(gx^*)(y) = \begin{cases} g \text{ if } x=y \\ 0 \text{ otherwise } \end{cases}
\end{align*}

Furthermore, these functions provide a basis for $G_{n-p}^{|K_n|}$ since any $f\in G_{n-p}^{|K_n|}$ can be written as $f = \sum_x f(x)x^*$. Therefore, we have:
\begin{align} \label{homeq}
\text{hom}(C,G)^p = \prod_n \text{Hom}(C_n,G_{n-p}) = \prod_n G_{n-p}^{|K_n|},
\end{align}
allowing to construct a basis for $\text{hom}(C,G)^p$ using the basis for $G_{n-p}^{|K_n|}$. To achieve this we extend $gx^*\in G_{n-p}^{|K_n|}$ to $gx^*\in \text{hom}(C,G)^p$ in the following way:
\begin{align*}
(gx^*)_m = \begin{cases}
gx^*, \text{if}\ m=n, \\
0, \text{if}\ m\neq n
\end{cases}
\end{align*}
and observe that any $j \in \text{hom}(C,G)^p$ can be written as:
\begin{align}\label{texp}
j = \sum_n \sum_{x \in K_n} j_n(x)x^*
\end{align}
so these maps provide a basis for $\text{hom}(C,G)_p$.

In a similar fashion, using the dualization procedure sketched in \ref{sec:math}, from equation (\ref{homeq}) we have:
\begin{align} 
\nonumber \text{hom}(C,G)_p & = \prod_n \hat{G}_{n-p}^{|K_n|} = \prod_n \text{Hom} \left( C_n, \hat{G}_{n-p} \right)  \\ 
\label{homeqdual} & = \text{hom} ( C, \hat{G} )_p,
\end{align}
where the last equality comes from equation (\ref{def:hom}). We will use mainly the first notation throughout the paper, unless some confusion arises. By the same token, given $r \in \hat{G}_{n-p}$ and $x\in K_n$ we obtain $rx_\ast \in \text{hom}(C,G)_p$ defined by:
\begin{align*}
r x_\ast (f) = r(f_n(x))
\end{align*}
and similarly, any $k \in \text{hom}(C,G)_p$ can be written as: 
\begin{align}\label{sexp}
k =\sum_n \sum_{x \in K_n} k_n(x)x_\ast \quad.
\end{align}
where $k_n(x)\in \hat{G}_{n-p}$ is defined by $k_n(x)(g) = s(gx^\ast)$.

As a consequence, we have that any generalized gauge transformation $A_t$ and holonomy $B_s$, can be decomposed as:
\begin{align*}
A_t & = A_{\sum_{n} \sum_{x \in K_n} t_n (x) x^*} = \prod_n \prod_{x\in K_n} A_{t_n(x)x^*} \quad, \\ 
B_m & = B_{\sum_{n} \sum_{x \in K_n} m_n (x) x_*} = \prod_n \prod_{x\in K_n} B_{m_n (x)x_*}.
\end{align*}

Moreover, we can use this decomposition to prove:

\begin{prop}\label{AtBsexp2}
Given $s \in \text{hom}(C,G)_{-1}$ and $ v \in \text{hom}(C,G)^{1}$, the following identities are satisfied:
\begin{align*}
\sum_t \chi_s (t) A_t & = \prod_n \prod_{x \in K_n} \sum_{g\in G_{n+1}}\chi_{s} \left( g x^* \right) A_{g x^*} \quad, \\ 
\sum_m \chi_m (v) B_m & = \prod_n \prod_{x \in K_n} \sum_{r\in \hat{G}_{n-1}} \chi_{r x_*} \left( v \right) B_{r x^*} \quad.
\end{align*}
\end{prop}

\begin{proof} 
Let us start by writing   $t := \sum_{n,x} g_x x^\ast$ and realizing that summing over all $t \in \text{hom}(C,G)^{-1}$ is equivalent to summing over all possible values of $g_x$, hence:
\begin{multline*}
\sum_t\chi_s(t)A_t  = \sum_{g_x}  \chi_s \left( \sum_{n, x \in K_n} g_x x^\ast\right) A_{\sum_{n,x \in K_n} g_x x^*} \\ 
=\sum_{g_x} \left( \prod_{n} \prod_{x \in K_n} \chi_{s} \left( g_x x^* \right)\right) \left(\prod_{n} \prod_{x \in K_n} A_{g_xx^*} \right) \\ 
= \sum_{g_x} \prod_{n} \prod_{x \in K_n}  \chi_{s} \left( g_x x^* \right) A_{g_xx^*} \\ 
= \prod_{n} \prod_{x \in K_n} \sum_{g\in G_{n+1}} \chi_{s} \left( g x^* \right) A_{gx^*}.
\end{multline*}

Similarly, let us write $m := \sum_{n,x} r_x x_*  \in \text{hom} (C,G)_{1}$. It is clear that summing over all $m \in \text{hom}(C,G)_{1}$ is equivalent to summing over all possible values of $r_x$, hence:
\begin{multline*}
\sum_m \chi_m (v) B_m  = \sum_{r_x}  \chi_{\sum_{n, x \in K_n} r_x x_\ast} \left( v\right) B_{\sum_{n, x \in K_n} r_x x_*} \\ 
= \sum_{r_x} \left( \prod_{n} \prod_{x \in K_n} \chi_{r_x x_*} \left( v \right) \right) \left(\prod_{n} \prod_{x \in K_n} B_{r_x x_*} \right) \\ 
= \sum_{r_x} \prod_{n} \prod_{x \in K_n}  \chi_{r_x x_*} \left( v \right) B_{r_x x_*} \\ 
= \prod_{n} \prod_{x \in K_n} \sum_{r \in \hat{G}_{n-1}} \chi_{r_x x_*} \left( v \right) B_{r_x x_*},
\end{multline*}
as was to be shown.
\end{proof}

From here, it is clear that we can also write the generalized projector operators as:
\begin{prop}[Generalized projector operators]\label{prop:ObOp}
Given $s \in \text{hom}(C,G)_{-1}, \, m \in \text{hom}(C,G)^1$ we can write:
\begin{align*}
\mathcal{A}_s := \prod_{n}\prod_{x \in K_n} A_{x}^{s_n \left( x \right)} \quad , \quad \mathcal{B}_m := \prod_{n}\prod_{x \in K_n} B_{x}^{m_n \left( x \right)} \quad,
\end{align*}
where we have used the decompositions (\ref{sexp}) and (\ref{texp}) for $s$ and $m$ respectively, as well as Definition \ref{def:AxBx}.
\end{prop}

\section{Supplementary Algebraic Relations}\label{ap:algrel}

The following lemma, follows immediately from the Definition \ref{def:PQoperators}, when applying the operators over a state $\ket{f}$ with $f \in \text{hom} (C,G)^0$, so we omit the proof:

\begin{lem}[Shift and Clock operators algebra]\label{lem:PQalgebra}
Let $t, t_1 \text{ and } t_2 \in \text{hom}\left( C, G \right)^0$ and $m, m_1 \text{ and } m_2 \in \text{hom} \left( C, G \right)_0$, the relations
\begin{align*}
P_{t_1} P_{t_2} = P_{t_1 + t_2} \quad &, \quad Q_{m_1} Q_{m_2} = Q_{m_1 + m_2} , \\ 
Q_m P_t  & = \chi_{m} \left( t \right) P_t Q_m ,
\end{align*}
hold.
\end{lem}

Using the previous lemma, and the Definitions \ref{def:AopBop}, we also have:

\begin{lem}[$A_t$ and $B_m$ algebra]\label{lem:ABalgebra}
Let $A_t$ and $B_m$ be as defined in \ref{def:AopBop}. They satisfy the following relations:
\begin{align*}
A_t A_{t'}=A_{t+t'}=A_{t'} A_t \,\, &, \,\, B_m B_{m'}=B_{m+m'}=B_{m'} B_m \,\,, \\ 
A_t B_m & = B_m A_t.
\end{align*}
\end{lem}

\begin{proof}
The first two relations are trivial, while for the third one it is immediate to compute:
\begin{align*}
B_m A_t  = \chi_{\delta_1 m}\left( \delta^{-1} t \right) P_{\delta^{-1} t} Q_{\delta_1 m} = \chi_m \left( \delta^0 \delta^{-1} t \right) A_t B_m,
\end{align*}
the result follows from the cochain property $\delta^0 \delta^{-1} = 0$. 
\end{proof}

Regarding the local operators we have the following set of relations:

\begin{lem}\label{lemma:AlgAB}
For all $x\in K_n,\, y\in K_m,\, r \in \hat{G}_{n+1},\, r' \in \hat{G}_{m+1}$ and $g \in G_{n-1}, g' \in G_{m-1}$, the following relations for the local gauge transformation and local gauge holonomy (Definition \ref{def:AxBx}) are satisfied:
\begin{enumerate}[(i)]
\item Pairwise commutation 
\begin{align*}
A_x^{r} A_y^{r'} = A_y^{r'} A_x^{r} \quad &; \quad B_x^{g} B_y^{g'} =B_y^{g'} B_x^{g}; \\ 
A_x^{r} B_y^{g} & = B_y^{g} A_x^{r}.
\end{align*}

\item Orthogonality
\begin{align*}
A_x^{r} A_x^{r'} = \delta(r,r')A_x^{r} \quad , \quad B_y^{g} B_y^{g'} =\delta(g, g') B_x^{g},
\end{align*}
where $\delta \left( \cdot,\cdot \right)$ is the Kronecker delta.

\item Completeness
\begin{align*}
\sum_{r \in \hat{G}_{n+1}} A^{r}_x = \mathbb{1}, \quad \quad \sum_{g \in G_{n-1}} B^{g}_x = \mathbb{1},
\end{align*}
where $\mathbb{1}$ is the identity operator.
\end{enumerate}
\end{lem}

\begin{proof}
\begin{enumerate}[(i)]
\item This set of relations are immediately satisfied by the pairwise commutation shown in Proposition \ref{ABalgebra}.

\item
It is straightforward to compute:
\begin{align*}
A_x^{r} A_x^{r'} & = \dfrac{ \sum_b \sum_c}{\left| G_{n+1} \right|^2}  \chi_{r} (b) \chi_{r'} (c) A_{bx^*}A_{cx^*} \\ 
& = \dfrac{\sum_b \sum_c}{\left| G_{n+1} \right|^2}  \chi_{r' - r}(c) \chi_{r}(b+c) A_{(b+c)x^*} \\ 
& =\dfrac{1}{\left| G_{n+1} \right|}  \sum_c \chi_{r' - r}(c) A_x^{r},
\end{align*}
from which the result follows, when using the orthonormal relations of the characters. The second expression is obtained analogously.

\item We readily compute:
\begin{align*}
\sum_{r \in \hat{G}_{n+1}} A_{x}^r & =\frac{1}{|G_{n+1}|} \sum_{g\in G_{n+1}} \sum_{r \in \hat{G}_{n+1}} \chi_{r} (g) A_{gx^*} \\ 
& = \sum_{g \in G_{n+1}} \delta \left( g, 0 \right) A_{gx^*} \\ 
& = A_{0 x^*} \equiv A_0
\end{align*}
from which the result follows, also when using the orthonormal relations of the characters. The second expression is obtained analogously.
\end{enumerate}
\end{proof}

Notice that the following commutation relations holds immediately from \ref{lemma:AlgAB}.

\begin{lem}[]
the Hamiltonian operator $H: \mathcal{H} \rightarrow \mathcal{H}$ as Defined in (\ref{def:Hamiltonian}), satisfy:
\begin{enumerate}[(i)]
\item $H A_x^{r} = A_x^{r} H$, for all $x \in K_n$ and $r \in \hat{G}_{n+1}$,
\item $H B^{g}_y = B^{g}_y H$, for all $y \in K_m$ and $g \in G_{m-1}$.
\end{enumerate}
where $A_x^{r}$ and $B_y^{g}$ are as in Definition \ref{def:AxBx}.
\end{lem}

\bibliographystyle{apsrev4-1}
\bibliography{sample.bib}
\end{document}